\documentclass[acmsmall,nonacm]{acmart}

\usepackage{microtype}
\usepackage{enumitem}
\usepackage{amsthm}
\usepackage{caption}
\usepackage{booktabs}
\usepackage{tikz}
\usepackage{xcolor}

\newcommand{\Hist}{\mathcal{H}}
\newcommand{\Spec}{\mathsf{Spec}}
\newcommand{\drule}{\mathrel{:\!-}}
\newcommand{\snap}{\mathsf{snap}}
\newcommand{\seal}{\mathsf{seal}}
\newcommand{\seq}{\mathbin{\triangleright}}
\newcommand{\Ord}{\preceq}
\newtheorem{definition}{Definition}
\newtheorem{theorem}{Theorem}
\newtheorem{lemma}{Lemma}
\newtheorem{corollary}{Corollary}
\theoremstyle{remark}
\newtheorem{remark}{Remark}
\newtheorem{example}{Example}
\newtheorem{observation}{Observation}

\numberwithin{definition}{section}
\numberwithin{theorem}{section}
\numberwithin{lemma}{section}
\numberwithin{corollary}{section}
\numberwithin{remark}{section}
\numberwithin{example}{section}
\numberwithin{observation}{section}
\numberwithin{assumption}{section}

\title{Determination Provenance: From Ambiguity to Algebra}
\author{Joseph M. Hellerstein}
\affiliation{%
  \institution{University of California, Berkeley and Amazon Web Services}
  \country{USA}
}
\date{\today}

\begin{document}

\begin{abstract}
  Many data systems admit multiple admissible outcomes for the same
  input: concurrent transactions may serialize in one of many orders; a logic
  program may have multiple stable models.
  Classical data provenance cannot even pose its question in such
  settings---it explains how a result was derived, but only after
  something has chosen which result to produce.
  We introduce \emph{determination provenance} to track the
  commitments that resolve this ambiguity.
  A tuple's \emph{support} is the set of resolutions under which it
  holds.
  Supports form a commutative semiring, and layered commitments induce
  a \emph{filtration} measuring each tuple's \emph{query-relative
  depth}---how many layers of semantic resolution it depends on.
  Positive relational algebra respects the filtration, enabling
  compositional robustness analysis and quantitative diagnosis of
  resolution cost.
  We instantiate the framework for transactional isolation and for
  $\mbox{Datalog}^\neg$; in both, classical semantic variants
  (isolation levels; negation semantics) correspond to different views
  of a single shared filtration.
\end{abstract}

\maketitle

\section{Introduction}
\label{sec:intro}
Algebraic data provenance explains the outputs of a deterministic query---a \emph{function}
from an input database to an output relation~\cite{green2007provenance,suciu2011probabilistic}. 
But production data systems are often based on a more ambiguous
specification---a \emph{relation} that maps an input scenario to multiple admissible outputs.
Examples include transactional databases (multiple acceptable schedules),
distributed systems (multiple consensus decisions), and even 
Datalog with negation (multiple stable models). 
Classical data provenance cannot account for this ambiguity; it cannot even
formulate a question about a particular outcome (e.g., why vs.\ why not)
because the choice of outcome is undetermined.

This paper lays the algebraic foundation for provenance over
relational specifications.
We introduce
\emph{determination provenance}, a structure that captures the
commitments that resolve semantic ambiguity.
A \emph{determination} is a layered sequence of irrevocable,
history-indexed
commitment events that constrains a relational specification until
it determines a single outcome for each input---the semantic
counterpart of what a correct implementation achieves on each run.
Classical provenance applies within each determination; determination
provenance tracks which commitments were needed and how query results
depend on them.

\begin{example}[Provenance under transactional ambiguity]
  \label{ex:example}
  Consider a relation $S(k,y)$ with initial contents $\{S(0,b)\}$.
  Three transactions execute concurrently with a query:
  \[
    T_1\!: \textsf{ins } S(1,b);\qquad
    T_2\!: \textsf{ins } S(2,d);\qquad
    T_3\!: \textsf{del } S(2,d);\qquad
    T_Q\!: Q(y) \drule S(k,y).
  \]
  All four are concurrent; any serialization is possible.
  Value $b$ appears in $Q$'s result regardless of ordering
  (the initial tuple $S(0,b)$ is never deleted): $b$ is \emph{robust}.
  Value $d$ is \emph{contingent}: it appears in serializations where
  $T_Q$ observes $T_2$'s insert before any effective delete---e.g.,
  $T_2 \prec T_Q \prec T_3$, or $T_3 \prec T_2 \prec T_Q$ (where the
  delete precedes the insert and is a no-op).
  It is absent when the query precedes the insert or follows an
  effective delete.
  Before resolving this ambiguity, the
  provenance question for $d$ is undetermined: under one determination
  the question is \emph{why}; under another it is \emph{why-not}.
  Section~\ref{sec:transactions} formalizes this example.
\end{example}

\noindent
Each concrete execution corresponds to one determination (the
commitments that the system actually made).
But $b$'s robustness and $d$'s contingency are properties of the
\emph{space} of all possible determinations, not of any single one.
To reason about such properties algebraically, we track each tuple's
\emph{support}: the set of determinations under which it holds.
Value $b$ has full support (every determination); $d$ has partial
support (the determinations in which $T_Q$
observes $T_2$'s insert).
Supports form a Boolean semiring (join intersects supports; union
merges them), and the familiar power of algebraic provenance transfers:
robustness reduces to checking whether support is full;
counterfactual questions are read from the support; quantitative
diagnosis from its measure.

Determination provenance provides richer structure than this flat
semiring, however.
When commitments are layered---some depending on others' outcomes---the
semiring carries a \emph{filtration}: a chain of sub-semirings
reflecting how much resolution each tuple requires.
A robust or impossible tuple has depth~$0$; a contingent tuple has
positive depth reflecting which layer it first depends on.

The key property on filtrations: positive RA cannot increase a result tuple's depth
beyond the maximum of its inputs.
This enables \emph{depth bounds} (queries over shallow tuples are
unaffected by deeper layers), \emph{semantic-change analysis}
(comparing specifications tuple-by-tuple), and a \emph{complexity
connection} (depth-$1$ coincides with PDB lineage; higher depths
correspond to layered PDB evaluation).
We instantiate the framework for transactional isolation and
$\mbox{Datalog}^\neg$; in both, classical semantic variants
(isolation levels; negation semantics) are different views
of a single filtration.

\paragraph{Contributions.}
Determination provenance provides new layered algebraic structure,
expressivity, and applicability:
\begin{enumerate}[label=(\roman*),nosep,leftmargin=*]
  \item The \emph{determination semiring} and its \emph{filtration}:
        supports over resolving determinations form a commutative
        semiring whose layered structure is reflected as a chain of
        sub-semirings that query evaluation respects
        (Sections~\ref{sec:det-semiring}--\ref{sec:filtration}).
  \item \emph{Instantiation for transactional systems}: isolation levels
        (RC, SI, SER) are different resolving subsets of a shared
        filtration; determination depth is $\Theta(n)$ worst-case
        (in the number of transactions) with scheduling discretion, and SER/SI are incomparable in
        per-tuple query-relative depth
        (Section~\ref{sec:transactions}).
  \item \emph{Instantiation for $\mbox{Datalog}^\neg$}: for programs
        with stable models, negation semantics (stratified,
        well-founded, stable) are different reading depths of a shared
        filtration; monus elimination shows that sealing is
        support-equivalent to semirings with monus for stratified
        programs; the determination framework also handles unstratified
        negation, where monus does not provide a single resolved
        semantics
        (Section~\ref{sec:negation-preview},
        Appendix~\ref{sec:appendix-datalog}).
  \item \emph{Quantitative measures and semantic-change analysis}:
        specifications over the same basis can be compared
        tuple-by-tuple, quantifying \emph{work regret}
        vs.\ genuine \emph{semantic shift}
        (Section~\ref{sec:discussion}).
  \item \emph{Complexity and attribution} (appendices):
        robustness is coNP-complete; a
        Shapley-value measure attributes contingency to individual
        commitments (Appendices~\ref{app:robustness-proofs},~\ref{app:responsibility-details}).
\end{enumerate}

\section{Preliminaries: Ambiguous Semantics and Refinement}
\label{sec:prelim}

We now formalize the setting: specifications with multiple admissible
outcomes and the irrevocable choice structures---called
\emph{determinations}---that restrict ambiguity until meaning is
determinate.

\subsection{Histories}
Histories are used to model ambiguity in distributed systems~\cite{lamport};
we use them here as well.
\begin{definition}[History]
  A \emph{history} is a finite partially ordered set
  $
    H = (E, \rightarrow),
  $
  where $E$ is a set of events and $\rightarrow \subseteq E \times E$ is a strict
  (irreflexive and transitive) partial order representing known precedence constraints among events.
\end{definition}

\noindent
A history records which events have occurred and their precedence.
The class $\Hist$ of all histories represents all possible states across
all possible runs.
When the partial order is total, a history is a sequence (e.g., a
transaction schedule).

To model how a history evolves, we define an extension relation. An extension may add new events and precedence edges involving them, but
may not revise the past of any existing event.

\begin{definition}[History Extension]
  For histories $H_1=(E_1,\rightarrow_1)$ and $H_2=(E_2,\rightarrow_2)$, we write
  $H_1 \sqsubseteq H_2$ if:
  \begin{itemize}[nosep,leftmargin=*,label={}]
    \item (1)~$E_1 \subseteq E_2$, 
          \quad and
          \quad (2)~$\rightarrow_1 = \rightarrow_2 \cap (E_1 \times E_1)$, 
          \quad and
          \quad (3)~no event in $E_1$ has a predecessor in $E_2 \setminus E_1$.
  \end{itemize}
\end{definition}

\subsection{Outcomes and Specifications}
An \emph{outcome} is an element of a set $O$ of externally
observable results (e.g., a query result, a database state, or a model of a
logic program).
The chosen outcome domain must include a partial order
$\Ord$~\cite{hellerstein2025coordinationcriterion}.
When $o_1 \Ord o_2$, replacing outcome $o_1$ with $o_2$ is not considered
to rule out $o_1$ but rather to \emph{refine} it---$o_2$ adds detail
without contradicting $o_1$.
The order $\Ord$ is part of the specification---a modeling choice, not
derived from execution structure. Changing $\Ord$ changes the
specification (enlarging $\Ord$ makes more outcomes compatible,
reducing the commitments needed to resolve ambiguity).

\begin{definition}[Specification]
  A \emph{specification} is a triple $(\Hist, O, \Ord)$ together with a
  relation $\Spec \subseteq \Hist \times O$
  associating each history with a set of \emph{admissible outcomes}.
  We write $\Spec(H) \triangleq \{o \in O \mid (H,o) \in \Spec\}$ for the
  admissible set at~$H$.
\end{definition}

\begin{example}[Specification for Example~\ref{ex:example}]
  \label{ex:spec}
  $O$ is the set of decision traces (commit/abort sequences for each
  transaction); $\Ord$ is trace-prefix extension (a longer trace
  refines a shorter one); $\Spec(H)$ is the set of traces consistent
  with some serialization of the concurrent transactions in $H$.
\end{example}

\noindent
As noted in the introduction, a specification is in general a
relation; the goal of determination is to turn it into a function
(a unique outcome per history).
Since provenance analysis is retrospective---the full history is
available---we define the \emph{resolved outcome} at a history $H$ as
the maximum element of $\Spec(H)$ under $\Ord$.

\begin{definition}[Determined specification]
  \label{def:determined}
  A specification is \emph{determined} at $H$ if $\Spec(H)$ forms a
  chain under $\Ord$ with a maximum element (all outcomes are pairwise
  comparable and a greatest one exists).
  For a determined history, we write $\Spec(H)$ as shorthand
  for $\max_{\Ord} \Spec(H)$ (the resolved outcome).
  A specification is \emph{determined} if it is determined at every
  $H \in \Hist$.
\end{definition}

\noindent
A determined specification defines a function from histories to
outcomes; classical provenance applies directly.
Ambiguity requiring determination arises only when $\Spec(H)$ contains
$\Ord$-incomparable outcomes (as in Example~\ref{ex:spec}, where
multiple serializations yield incomparable decision traces).

\subsection{Commitments}

Ambiguity is constrained via additional semantic commitments.
Commitments are operators that narrow the admissible set;
applying a commitment at a specific history produces a
\mbox{\emph{commitment event}.}

\begin{definition}[Commitment and commitment event]
  \label{def:commitment}%
  A \emph{commitment} $\varphi$ is an operator that can be applied at a
  history $H$ to produce an extended history $H \cdot \varphi$ with
  $\Spec(H \cdot \varphi) \subseteq \Spec(H) \wedge \Spec(H \cdot \varphi) \neq \emptyset$.
  The effect of $\varphi$---which outcomes it excludes---depends on the
  history at which it is applied; a commitment has no
  history-independent semantics.
  The resulting event in $H \cdot \varphi$ is a \emph{commitment event}:
  the irrevocable record that $\varphi$ was applied at this point.
\end{definition}

\noindent
Here $H \cdot \varphi$ denotes the history obtained by appending
the commitment event after all maximal events of $H$ (giving
$H \sqsubseteq H \cdot \varphi$).
Most commitments strictly shrink the admissible set; a
\emph{sealing commitment} $\varphi_{\seal(S)}$ may leave it unchanged,
merely declaring a set of events $S$ complete and enabling subsequent
commitments that depend on that completeness.

Because a commitment event succeeds all maximal events, it induces a
\emph{cut} in the history poset, ensuring that commitment events form a
total order within any history.
Applied sequentially, commitments monotonically narrow the admissible
set:
$\Spec(H) \supseteq \Spec(H \cdot \varphi_1) \supseteq
\Spec(H \cdot \varphi_1 \cdot \varphi_2) \supseteq \cdots$.
In general, the effect of $\varphi$ may depend on the
full admissible set at its point of application, not only on individual
outcomes---this is what produces non-commutative commitments.

\begin{definition}[Commitment basis]
  \label{def:commitment-basis}
  A \emph{commitment basis} $\Phi$ is a set of commitments (operators)
  from which determinations are formed.
  Different bases for the same specification yield different determination
  structures and different provenance.
\end{definition}

\noindent
The choice of commitment basis is a modeling decision as fundamental as
the choice of outcome space ($O$, $\Ord$).
Given an ambiguous specification (a relation), implementations select a
single outcome on each run; the commitment basis lifts the strategies
that such implementations use into a semantic domain, enabling
reasoning across different implementations and their nondeterminism.
(In Example~\ref{ex:example}, $\Phi = \{\varphi_{T_i \prec T_j}\}$:
the set of all pairwise ordering operators, capturing shared
aspects of 2PL, OCC and MVTO.)

\paragraph{Three forces on a history.}
\label{sec:three-forces}%
A history evolves under three
forces~\cite{hellerstein2025complexity}:
\emph{environment events} extend the history but are not chosen;
\emph{commitments} replace the admissible set with a subset
(irrevocably constraining which outcomes remain possible);
\emph{entailments} refine the observed outcome upward in $\Ord$
without excluding any alternative (they add detail that is compatible
with every remaining admissible outcome).
In transactional systems, operation interleavings are environment and
commit/abort decisions are commitments; in $\mbox{Datalog}^\neg$, EDB
facts are environment, sealing and choice predicates are commitments,
and the alternating fixpoint's classifications are entailments.

\subsection{Determinations}

Individual commitments typically resolve only part of the ambiguity.
A full resolution requires a collection of them:

\begin{definition}[Determination]
  \label{def:determination}%
  The \emph{determination} of a history $H$ over a commitment basis
  $\Phi$ is the subsequence of commitment events in $H$, listed in the
  order induced by $\rightarrow$:
  $D(H) = \varphi_1 \cdot \varphi_2 \cdot \cdots \cdot \varphi_m$.
  Non-commitment events may occur between successive commitments; the
  effect of each $\varphi_i$ depends on the full history at its point
  of application.
\end{definition}

\noindent
A history $H$ is \emph{resolved} if $\Spec(H)$ is determined
(Definition~\ref{def:determined}).
Note that $\Spec$ itself is a fixed relation---it does not change as
history grows.
Commitment events can only narrow: $\Spec(H \cdot \varphi) \subseteq
\Spec(H)$ by definition.
Non-commitment events may change the admissible set in either direction;
they provide the context in which commitments take effect (and hence
why determinations are history-indexed).
When the history is clear from context we write $D$ for $D(H)$.

\begin{example}[Determinations for Example~\ref{ex:example}]
  \label{ex:det-intro}
  The ordering decisions $\varphi_{T_2 \prec T_Q}$ and
  $\varphi_{T_Q \prec T_3}$ are commitments: each excludes
  serializations incompatible with the chosen order.
  The determination $D_{\mathsf{in}} = \varphi_{T_2 \prec T_Q} \cdot
  \varphi_{T_Q \prec T_3}$ resolves the specification so that $T_Q$
  observes $T_2$'s insert; $D_{\mathsf{out}} = \varphi_{T_2 \prec T_3}
  \cdot \varphi_{T_3 \prec T_Q}$ resolves it so that $d$ is absent.
\end{example}

\paragraph{Layer-sequencing notation.}
Commitments in a determination are sequenced for two reasons.
First, non-commitment events may intervene: in
$H \cdot \varphi \cdot E \cdot \psi$, the commitment $\psi$ may depend
on event $E$, so $\varphi$ and $\psi$ cannot be reordered.
Second, even contiguous commitment events may not commute:
$\Spec(H \cdot \varphi \cdot \psi) \neq \Spec(H \cdot \psi \cdot \varphi)$
when one commitment's effect depends on the other's exclusions.
A \emph{layer} is a maximal set of contiguous commitment events (no
intervening non-commitment events) that commute pairwise.
We use $\seq$ to denote a sequence of such layers:

\begin{definition}[Sequencing of layers (notation)]
  \label{def:layer-seq}
  Let $L_1,\ldots,L_k$ be multisets of commitments.
  We write $D = L_1 \seq L_2 \seq \cdots \seq L_k$ to indicate that
  determination $D$'s commitments are organized into $k$ layers, applied
  in order with non-commitment events possibly intervening between
  layers.
  Within each $L_i$, commitments commute; across layers, order
  matters.
\end{definition}

\noindent
The \emph{depth} of a determination $L_1 \seq \cdots \seq L_k$ is $k$
(the number of layers).

\begin{definition}[Resolving determination]
  A determination $D(H)$ is \emph{resolving} if $H$ is resolved
  (Definition~\ref{def:determined}).
\end{definition}

\begin{definition}[Minimal resolving determination]
  \label{def:minimal-resolving}
  A resolving determination $D(H)$ is \emph{minimal} if removing any
  single commitment event from $H$ (while retaining all other events
  in their original order) yields a history that is not resolved.
\end{definition}

\subsection{Resolution Enables Provenance}
\label{sec:mono-theorem}

Classical semiring provenance annotates base tuples and propagates
annotations through query evaluation~\cite{green2007provenance}.
This applies directly to any resolved history: the determination
selects a concrete outcome, and standard evaluation traces its
derivation.
The question is what happens when the specification is \emph{not}
resolved---when incompatible outcomes coexist.

\begin{theorem}[Classical provenance is pointwise in determinations]
  \label{thm:resolution-implies-monotone}
  Classical semiring provenance (a single $K$-relation) correctly
  represents a specification $\Spec$ iff all resolving determinations
  produce the same tuple membership and derivational annotation.
  When two resolving determinations $D_1, D_2$ disagree on some tuple
  $t$ ($P_{D_1}(t) \neq P_{D_2}(t)$), any sound representation must
  distinguish $D_1$ from $D_2$.
\end{theorem}

\begin{proof}
  Suppose two resolving determinations $D_1, D_2$ disagree: some tuple
  $t$ satisfies $P_{D_1}(t) \neq P_{D_2}(t)$ (where inequality
  includes one being $0$).
  Since semiring provenance expresses only positive dependence on base
  facts~\cite{green2007provenance}, no single annotation can
  explain both presence under $D_1$ and absence under $D_2$.
  Hence any sound representation must distinguish $D_1$ from $D_2$.
  Conversely, if all resolving determinations agree, the common
  $K$-relation is sound and complete for all of them.
\end{proof}

\section{The Algebra of Determinations}
\label{sec:conditional-semiring}

Classical provenance owes its power to algebraic structure: derivations
compose through a commutative semiring, and query evaluation preserves
that structure~\cite{green2007provenance}.
One might expect that layered determinations---sequences of non-commuting
commitments---resist such clean algebraic treatment.
In fact, supports over resolving determinations form a commutative
semiring, and the layering within each determination is reflected as a
filtration of sub-semirings that query evaluation respects.

Throughout, we fix a specification $\Spec$ over a history class $\Hist$ and
a commitment basis $\Phi$.
Elements of $\mathcal{D}_{\Spec,\Phi}$ (abbreviated $\mathcal{D}$) are
resolved histories with minimal determinations, identified up to
outcome equivalence: $H_1 \sim H_2$ when $D(H_1)$ and $D(H_2)$ produce
the same resolved outcome and the same conditioned provenance for every
tuple.
Concretely, each element of $\mathcal{D}$ is a representative resolved
history $H$ whose determination $D(H)$ is minimal
(Definition~\ref{def:minimal-resolving}).
Unlike abstract possible worlds, elements of $\mathcal{D}$ are
grounded in event structure: the same commitment may have different
effects at different histories, so determinations are history-indexed
records, not symbolic labels.

\subsection{The Determination Semiring}
\label{sec:det-semiring}
\label{sec:algebra}

Fix a commutative semiring $(K, +, \cdot, 0, 1)$ for derivational provenance
(e.g., the polynomial semiring $\mathbb{N}[X]$ of~\cite{green2007provenance}).
A \emph{$K$-relation} assigns each tuple a value in $K$: alternative
derivations combine by $+$, joint use of facts by $\cdot$, absence by $0$.
For a determination $D \in \mathcal{D}$ extracted from a resolved
history $H$, the \emph{conditioned provenance} $P_D(t) \in K$ is the
classical semiring provenance of tuple $t$ computed over the resolved
outcome of $H$.

Conditioned provenance suffices for explaining a single execution.
But some questions require reasoning \emph{across} determinations:
Does this result hold under every admissible commit ordering, or is it
contingent on a particular one?
Answering such questions requires comparing provenance across multiple
determinations.

\paragraph{Observables and provenance queries.}
To connect outcomes to provenance, we fix a relational schema
$\mathbf{S}$ and an \emph{observation function}
$\mathsf{obs}: O \to \mathrm{Inst}(\mathbf{S})$ mapping each outcome
to a database instance.
In our $\mbox{Datalog}^\neg$ examples, $\mathsf{obs}$ is the identity
(outcomes are models, i.e., database instances); for transactions,
$\mathsf{obs}$ maps a decision trace to the database state that
results from applying committed transactions' writes.
A \emph{provenance query} is a positive relational algebra expression
over $\mathbf{S}$; conditioned provenance is the familiar
$K$-relational evaluation of Green et
al.~\cite{green2007provenance} applied to $\mathsf{obs}(\Spec(H))$.
This is where the framework gains expressive power beyond
per-commitment analysis: provenance queries can select individual
tuples or join across transactions --- enabling fine-grained and
cross-transaction sensitivity analysis.
Throughout, ``tuple'' refers to a tuple in $\mathsf{obs}(o)$ or in a
query result over it.

\begin{definition}[Determination Provenance]
  \label{def:det-provenance}
  Fix a specification $\Spec$ and a provenance query $Q$ over
  $\mathbf{S}$.
  Each determination $D \in \mathcal{D}$ corresponds to a resolved
  history $H$; evaluating $Q$ over $\mathsf{obs}(\Spec(H))$ yields a
  $K$-annotation for each tuple $t$.
  The \emph{determination provenance} of $t$ is the function
  \[
    P_{\Spec}(t) : \mathcal{D} \to K
  \]
  mapping each determination to the conditioned provenance $P_D(t)$
  ($= 0_K$ when $t$ is absent).
  When $\Spec$ is clear from context, we write $P(t)$.
\end{definition}

\noindent
Determination provenance is an element of the product semiring
$K^{\mathcal{D}}$ (pointwise addition and multiplication):
it assigns each $D \in \mathcal{D}$ a semiring value
recording how $t$ is derived under $D$.
Conditioned provenance is the special case of evaluating at a single determination:
$P(t)(D) = P_D(t)$.
Classical provenance is the special case where $|\mathcal{D}| = 1$
(the specification is already determined; no commitments needed).

\begin{definition}[Support and the determination semiring]
  The \emph{support} of a tuple $t$ is
  $\mathrm{supp}(P(t))
    \triangleq
    \{\, D \in \mathcal{D} \mid P(t)(D) \neq 0 \,\}$:
  the set of determinations under which $t$ holds.
  Supports form a Boolean algebra
  $(2^{\mathcal{D}}, \cup, \cap, \emptyset, \mathcal{D})$
  under the operations induced by query evaluation (join intersects
  supports; union merges them).
\end{definition}

\noindent
This algebra is elementary; the non-trivial structure is the
\emph{filtration} (Section~\ref{sec:filtration}), which reflects the
layered commitment process.


\begin{example}[Determination provenance for Example~\ref{ex:example}]
  \label{ex:det-semiring}
  Two resolved histories arise from Example~\ref{ex:example}: one in
  which $T_2$'s insert is visible to $T_Q$ (determination
  $D_{\mathsf{in}}$) and one in which $T_3$'s delete precedes $T_Q$
  (determination $D_{\mathsf{out}}$).
  With $\mathcal{D} = \{D_{\mathsf{in}}, D_{\mathsf{out}}\}$:
  \[
    P(b) = \{(D_{\mathsf{in}},\, x_0),\; (D_{\mathsf{out}},\, x_0)\}
         = \mathcal{D} \times \{x_0\},
    \qquad
    P(d) = \{(D_{\mathsf{in}},\, x_2)\}.
  \]
  Here $x_0$ annotates the initial tuple $S(0,b)$ and $x_2$ annotates
  $T_2$'s insert of $S(2,d)$.
  $P(b)$ has full support ($\mathcal{D}$): $b$ is robust.
  $P(d)$ has partial support ($\{D_{\mathsf{in}}\}$): $d$ is contingent.
  Composition operates on both components: a join intersects supports
  (multiplies $K$-values pointwise); a union takes their union (adds
  $K$-values pointwise).
\end{example}

\paragraph{Computational cost.}
When $\mathcal{D}$ is finite, $|\mathcal{D}|$ may be exponential in the
number of commitments (each binary commitment doubles the space).
Naively enumerating supports is therefore impractical.
For single-layer determinations, supports admit compact representation
as positive Boolean formulas over commitment variables
(Proposition~\ref{prop:collapse-depth-1}).

\subsection{Filtration of the Determination Semiring}
\label{sec:filtration}

The Boolean algebra on supports tells us \emph{whether} a tuple is
contingent, but not \emph{why}---which commitment layer it depends on.
Every determination $D$ is a sequence of commitment events
(Definition~\ref{def:determination}); within that sequence, a
\emph{layer} is a maximal contiguous set of commitments that commute
pairwise in the history at which they are applied.
The layers of $D$ are its Foata normal
form~\cite{mazurkiewicz1977concurrent}: the coarsest grouping into
commuting stages.
(In general, commutativity may be dynamic---dependent on the current
admissible set---making canonical layering subtle; see
Appendix~\ref{app:open-questions}.
In both instantiations below, commutativity is history-independent,
so the layer structure is unambiguous.)
Write $L_1(D), L_2(D), \ldots$ for the successive layers of $D$;
the \emph{depth} of $D$ is its number of layers.

\begin{definition}[Level-$k$ agreement]
  \label{def:level-k-agreement}
  Two determinations $D, D' \in \mathcal{D}$ \emph{agree at level $k$},
  written $D \equiv_k D'$, if they have the same history through
  layer~$k$ (same events, same commitments, same order, including all
  intervening non-commitment events).
\end{definition}

\noindent
Each $\equiv_k$ is an equivalence relation; classes refine as $k$
increases.

\begin{definition}[Filtration]
  \label{def:filtration}
  A set $S \subseteq \mathcal{D}$ is a \emph{level-$k$ support} if it
  is a union of $\equiv_k$ classes (closed under level-$k$ agreement).
  Define $\mathcal{F}_k \triangleq \{S \subseteq \mathcal{D} \mid
  S \text{ is a level-}k\text{ support}\}$.
  This yields a filtration:
  $\mathcal{F}_0 \subseteq \mathcal{F}_1 \subseteq \cdots$, with
  $\mathcal{F}_0 = \{\emptyset, \mathcal{D}\}$ and
  $\mathcal{F}_d = 2^{\mathcal{D}}$ once all determinations are
  distinguished.
\end{definition}

\noindent
A tuple with qdepth~$0$ is robust (full support) or impossible (empty
support); positive qdepth reflects which layer first differentiates
the tuple's membership across determinations.

\begin{observation}[Shared filtration across specifications]
  \label{obs:basis-entails-filtration}
  When a domain's conflict structure determines a natural commitment
  basis $\Phi$ (ordering and abort operators for transactions;
  sealing and choice operators for $\mbox{Datalog}^\neg$), different
  specifications over that basis share a single filtration.
  For transactions, isolation levels (RC, SI, SER) are different
  resolving subsets within one filtration.
  For $\mbox{Datalog}^\neg$, negation semantics (stratified,
  well-founded, stable) are different reading depths of one
  filtration.
  In both cases, the filtration is a property of the workload's
  conflict or dependency structure, not of the policy imposed on it.
\end{observation}

\noindent
For cross-specification comparison, we fix an ambient determination
space $\mathcal{D}^\star_\Phi$: the set of resolved-history classes
generated by the basis before imposing any particular specification's
admissibility predicate.
Each specification $I$ selects a resolving subset
$\mathcal{D}_I \subseteq \mathcal{D}^\star_\Phi$; we extend each
support by zero outside $\mathcal{D}_I$, so all supports are subsets
of the common carrier $\mathcal{D}^\star_\Phi$ and qdepth is compared
over this shared space.

\begin{remark}[Uniform layering in the instantiations]
  \label{rem:uniform-layering}
  In both instantiations of this paper, commutativity is
  history-independent (it holds universally or not at all), so all
  determinations share the same layer structure: for transactions,
  independent ordering decisions form a single layer while
  cycle-breaking decisions form subsequent layers; for
  $\mbox{Datalog}^\neg$, stratum sealing forms layers $1, \ldots, k$
  and choice predicates form layer $k{+}1$.
  We exploit this uniform structure in the instantiation sections
  without further comment.
\end{remark}

\noindent
$\mathcal{F}_0$ captures the coarsest distinction: $\mathcal{D}$ is the
support of a robust outcome (holds under every determination) and
$\emptyset$ is the support of an impossible one (holds under none).
No layer information is needed to make this distinction.

\begin{proposition}[Filtration respects the determination semiring]
  \label{prop:filtration-semiring}
  Each $\mathcal{F}_k$ is closed under $\cup$ and $\cap$, hence
  $(\mathcal{F}_k, \cup, \cap, \emptyset, \mathcal{D})$ is a
  sub-semiring of the determination semiring for each $k$.
\end{proposition}
\begin{proof}
  Unions and intersections of unions-of-equivalence-classes are again
  unions-of-equivalence-classes.
\end{proof}

\noindent
Concretely, the filtration gives each support the structure of a
\emph{trie} (layer~$k$ branches conditioned on layers $1,\ldots,k{-}1$).
We discuss the analogy to $\mathbb{N}[X]$ universality---and its
limits---in Appendix~\ref{app:open-questions}.

\begin{proposition}[Single-layer case]
  \label{prop:collapse-depth-1}
  When all commitments commute (a single layer), the filtration has
  exactly two levels ($\mathcal{F}_0 = \{\emptyset, \mathcal{D}\}$ and
  $\mathcal{F}_1 = 2^{\mathcal{D}}$), and supports coincide with
  $\mathrm{PosBool}(\Phi)$---positive Boolean formulas over commitment
  variables, interpreted over valid minimal determinations.
  This connects to PDB lineage in Section~\ref{sec:discussion}.
\end{proposition}
\begin{proof}
  With a single layer, two determinations agree at level $1$ iff they
  apply the same set of commitments---that is, iff they are the same
  determination.
  The level-$1$ equivalence classes are therefore singletons $\{D\}$.
  A level-$1$ support is any union of such singletons (by definition),
  which is an arbitrary subset of $\mathcal{D}$; hence
  $\mathcal{F}_1 = 2^{\mathcal{D}}$.
  Each determination corresponds to a conjunction of commitment variables
  (which commitments were applied); supports correspond to disjunctions
  of such conjunctions---exactly $\mathrm{PosBool}(\Phi)$.
\end{proof}

\begin{definition}[Query-relative depth]
  \label{def:query-depth}
  The \emph{query-relative depth} of a tuple $t$ is
  $\mathrm{qdepth}(t) \triangleq
  \min\{k \mid \mathrm{supp}(P(t)) \in \mathcal{F}_k\}$.
\end{definition}

\begin{proposition}[Characterization]
  \label{prop:qdepth-characterization}
  \leavevmode
  \begin{enumerate}[nosep,label=(\alph*)]
    \item $\mathrm{qdepth}(t) = 0$ iff $t$ is robust
          ($\mathrm{supp} = \mathcal{D}$) or impossible
          ($\mathrm{supp} = \emptyset$).
    \item $\mathrm{qdepth}(t) = k$ iff all determinations agreeing on
          layers $1, \ldots, k$ agree on whether $t$ holds, but some
          pair agreeing on layers $1, \ldots, k{-}1$ disagrees.
    \item $\mathrm{qdepth}(t) = d$ iff $t$'s presence depends on the
          full determination.
  \end{enumerate}
\end{proposition}
\begin{proof}
  Immediate from the definitions of $\mathcal{F}_k$ and query-relative depth.
\end{proof}

\begin{example}[Query-relative depth for transactions]
  \label{ex:qdepth-txn}
  In Example~\ref{ex:example}, the conflict graph is acyclic, so all
  transaction-commit ordering commitments commute and the determination
  has a single layer.
  Value $b$ has $\mathrm{qdepth}(b) = 0$: it holds under every
  determination (robust).
  Value $d$ has $\mathrm{qdepth}(d) = 1$: its support is a proper subset
  of $\mathcal{D}$ (determinations where $T_Q$ observes $T_2$'s insert),
  which is not in $\mathcal{F}_0 = \{\emptyset, \mathcal{D}\}$.
  Classical provenance works for $b$; determination provenance is needed
  for $d$.
\end{example}

\begin{corollary}[Query evaluation respects the filtration]
  \label{cor:propagation-filtration}
  Under positive relational algebra:
  $\mathrm{qdepth}(t_1 \bowtie t_2) \le
  \max(\mathrm{qdepth}(t_1), \mathrm{qdepth}(t_2))$
  and
  $\mathrm{qdepth}(t_1 \cup t_2) \le
  \max(\mathrm{qdepth}(t_1), \mathrm{qdepth}(t_2))$.
  Query evaluation cannot increase query-relative depth beyond the maximum
  depth of its inputs.
\end{corollary}
\begin{proof}
  Join computes $\mathrm{supp}(t_1) \cap \mathrm{supp}(t_2)$;
  union computes $\mathrm{supp}(t_1) \cup \mathrm{supp}(t_2)$.
  By Proposition~\ref{prop:filtration-semiring}, both remain in
  $\mathcal{F}_{\max(k_1, k_2)}$.
  Selection by a determination-independent predicate either preserves
  a tuple's support or removes the tuple entirely, hence preserves
  filtration membership.
  Projection unions supports over matching tuples; since
  $\mathcal{F}_k$ is closed under union, projection preserves
  filtration membership.
\end{proof}

\section{Transactions as Semantic Ambiguity}
\label{sec:transactions}

In transactional systems, semantic ambiguity arises from
\emph{conflict resolution}: when concurrent transactions conflict,
the system must decide which effects persist.
We show that any isolation level forbidding a conflict cycle has
worst-case determination depth $\Theta(n)$ (in the number of
transactions) given scheduling
discretion---a result that holds regardless of concurrency control
protocol.
The filtration refines the classical incomparability of SER and
SI~\cite{adya1999weak} to a per-tuple depth comparison, and extends
per-transaction portability analysis~\cite{vandevoort2025mixed} to
arbitrary provenance queries over the resulting database.

\subsection{Transactional Histories}

\begin{definition}[Transactional history]
  \label{def:txn-history}
  A \emph{transactional history} is a history $H = (E, \rightarrow)$ whose
  event set $E$ contains:
  \begin{itemize}[nosep]
    \item $\mathsf{begin}(T_i)$, $\mathsf{commit}(T_i)$, $\mathsf{abort}(T_i)$
          --- lifecycle events for transaction $T_i$;
    \item $\mathsf{r}(T_i, x, v)$, $\mathsf{w}(T_i, x, v)$ --- read and write
          operations by $T_i$ on object $x$ with value $v$.
  \end{itemize}
  The partial order $\rightarrow$ records precedence: within a
  transaction, operations are totally ordered; across transactions,
  $e_1 \rightarrow e_2$ when $e_1$ is known to precede $e_2$ (e.g., a
  write that a subsequent read observes).
  Under protocols with scheduling discretion, ordering decisions among
  conflicting operations are the commitment events
  (Section~\ref{sec:scheduling-basis}); commit and abort outcomes
  follow from the chosen orderings.
  All other events (arrivals, reads, writes) are non-commitment events;
  they extend the history and may introduce conflict edges but do not
  themselves exclude outcomes.
  Aborted transactions remain in $E$: an abort records that a conflict
  was resolved against that transaction.
\end{definition}

\begin{definition}[Conflict]
  \label{def:conflict}
  Two operations from distinct transactions $T_i \neq T_j$ \emph{conflict} if
  they access the same object $x$ and at least one is a write.
  Following Adya~\cite{adya1999weak}, we distinguish three conflict types:
  \emph{write-write} ($\mathsf{ww}$: both write $x$),
  \emph{write-read} ($\mathsf{wr}$: $T_i$ writes $x$, $T_j$ reads $x$), and
  \emph{read-write} ($\mathsf{rw}$: $T_i$ reads $x$, $T_j$ writes $x$;
  also called an \emph{anti-dependency}).
\end{definition}

\begin{definition}[Conflict graph]
  \label{def:conflict-graph}
  The \emph{conflict graph} $G(H)$ derived from a transactional
  history $H$ is a directed graph whose vertices are the \emph{active}
  (uncommitted, non-aborted) transactions in $H$.
  An edge $T_i \xrightarrow{\ell} T_j$ (labeled
  $\ell \in \{\mathsf{ww}, \mathsf{wr}, \mathsf{rw}\}$) exists
  whenever $T_i$ and $T_j$ conflict on some object with conflict type
  $\ell$ and $T_i$'s conflicting operation precedes $T_j$'s in
  $\rightarrow$.
  For an extension $H \sqsubseteq H'$, the graph $G(H')$ omits from
  $G(H)$ transactions that committed or aborted in $H'$.
\end{definition}

\subsection{Commitment Basis and Specifications}
\label{sec:scheduling-basis}

Standard concurrency-control protocols all effectively order
conflicting transactions: OCC
via validation order, MVTO via timestamp assignment, 2PL via the time of first lock release
or deadlock-driven abort.
We capture this shared structure with an \emph{ordering commitment
basis}:

\begin{definition}[Ordering commitments]
  \label{def:commit-commitment}
  The ordering basis is
  $\Phi = \{\varphi_{T_i \prec T_j} \mid T_i, T_j \text{ conflict}\}$,
  where $\varphi_{T_i \prec T_j}$ irrevocably records that $T_i$ is
  serialized before $T_j$.
  \label{def:ordering-commitment}%
\end{definition}

\noindent
The transactional specification appeared in Example~\ref{ex:spec}:
$O$ is the set of decision traces (commit/abort
events, partially ordered by conflict edges), and $\Ord$ is
sub-trace inclusion.
Following Adya~\cite{adya1999weak}, each isolation level $L$
constrains which labeled cycle patterns are admissible;
$\Spec_L(H)$ admits extensions consistent with $L$'s constraint.
We consider:
\emph{read committed} (RC; no constraint),
\emph{serializability} (SER; forbids all directed cycles regardless of
edge labels), and \emph{snapshot isolation} (SI; permits rw-only
cycles but enforces first-committer-wins on $\mathsf{ww}$ edges;
Appendix~\ref{sec:appendix-SI}).

\begin{example}[Running example: depth~$1$ under SER]
  \label{ex:running-revisited}
  Returning to Example~\ref{ex:example}: the conflict graph is
  acyclic, so all ordering commitments commute and form a single layer.
  Two query-outcome equivalence classes of determinations arise
  (we write representative members):
  $D_{\mathsf{in}} \triangleq
    \varphi_{T_2 \prec T_Q} \cdot \varphi_{T_Q \prec T_3}$
  (all orderings where $T_Q$ observes $T_2$'s insert)
  and
  $D_{\mathsf{out}} \triangleq
    \varphi_{T_2 \prec T_3} \cdot \varphi_{T_3 \prec T_Q}$
  (orderings where $d$ is absent).
  The determination provenance is:
  \[
    P(b) = \mathcal{D} \times \{x_0\}, \qquad
    P(d) = \{(D_{\mathsf{in}},\, x_2)\}.
  \]
  $b$ is robust ($\mathrm{qdepth} = 0$); $d$ is contingent
  ($\mathrm{qdepth} = 1$).
  Figure~\ref{fig:det-tables} summarizes the determination structure
  across isolation levels.
\end{example}

\subsection{Determination Depth by Isolation Level}

\begin{theorem}[Determination depth for transactions]
  \label{thm:txn-depth}
  The following bounds depend only on the protocol's resolution class:
  \begin{enumerate}[nosep,label=(\alph*)]
    \item Under any isolation level that forbids no cycle type (e.g.,
          read committed): depth~$0$.
          No conflict resolution is required to satisfy the isolation
          specification.
    \item Under any isolation level $L$ that forbids some cycle type
          (serializability, snapshot isolation, repeatable read, etc.),
          with a fully reactive protocol (no scheduling discretion):
          depth~$0$.
          Every response is entailed by the current state and the
          arriving request.
    \item Under any such $L$ with scheduling discretion (batching,
          commit-order choice, or victim selection), whose forbidden
          cycle pattern can be generated repeatedly around a surviving
          discretionary transaction: worst-case
          depth~$\Theta(n)$, where $n$ is the number of transactions.
          Per-batch depth is~$2$ (seal the batch; choose a processing
          order within it).
  \end{enumerate}
\end{theorem}
\begin{proof}[Proof sketch]
  (a)~If no cycle type is forbidden, validation always passes and no
  transaction need abort; all outcomes are $\Ord$-comparable.

  (b)~A fully reactive protocol processes each event deterministically
  given the current state.
  Witness: textbook OCC (validate on commit request; abort iff cycle exists) or
  textbook MVTO (assign timestamp on arrival; force abort on violation).
  No system choice arises; depth~$0$.

  (c)~\emph{Upper bound}: each commitment resolves at least one
  transaction's fate; with $n$ transactions, at most $n$ sequential
  commitments suffice.
  \emph{Lower bound}: a long-running transaction $T_\infty$ (reading
  object $x$, writing object $y$) forms a two-edge rw-rw cycle with
  each fresh arrival $T_i$ (which writes $x$ and reads $y$):
  $T_\infty \xrightarrow{\mathsf{rw}} T_i$ and
  $T_i \xrightarrow{\mathsf{rw}} T_\infty$.
  Each round forces a binary decision: abort $T_\infty$ (resolving all
  future conflicts) or abort $T_i$ (letting $T_\infty$ survive to the
  next round).
  If the system keeps $T_\infty$ alive for $n{-}1$ rounds, each
  round's decision depends on the previous (had $T_\infty$ been
  aborted earlier, no further conflict would arise), giving
  depth~$n{-}1$.
  \emph{Per-batch}: within a sealed batch, the system makes a single
  commitment: the choice of processing order $\pi$ (a total order on
  pending transactions).
  All consequences (conflict edges, aborts, commits) are entailed by
  $\pi$.
  Hence per-batch depth is~$2$: seal + order selection.
  (Under 2PL the second layer is victim selection rather than order
  selection, but the seal is still needed to prevent new cycles
  from forming after resolution.  Details in
  Appendix~\ref{app:protocols}.)
\end{proof}

\begin{example}[Per-batch depth under OCC with batch validation]
  \label{ex:overlapping-cycles}
  Three transactions $T_1, T_2, T_3$ execute concurrently and complete
  their read/write phases.
  The system seals the batch (layer~1): no further transactions will
  join this validation round.
  It then chooses a validation order $\pi$ (layer~2) --- say
  $T_1 \prec T_2 \prec T_3$.
  Under $\pi$, $T_2$'s write set is checked against $T_1$'s; if they
  conflict, $T_2$ aborts (entailed by $\pi$, not a separate
  commitment).
  Choosing $T_2 \prec T_1 \prec T_3$ instead might let $T_2$ commit
  and abort $T_1$.
  The two layers do not commute: the seal determines which transactions
  are in the batch; the order determines which survive validation.
\end{example}

\subsection{Per-Tuple Isolation Sensitivity}

Vandevoort et al.~\cite{vandevoort2025mixed} certify individual
transactions as isolation-insensitive --- a per-transaction result.
Provenance queries over the resulting database enable a richer cut:
both finer grain (a query selecting a single tuple) and
cross-transaction questions (a join combining outputs of multiple
transactions).
The filtration's compositionality
(Corollary~\ref{cor:propagation-filtration}) guarantees that
insensitivity propagates through any such query.

\begin{proposition}[SER/SI qdepth incomparability]
  \label{prop:ser-si-incomparable}
  In the Adya conflict-graph model, the following two patterns witness
  the two directions of SER/SI qdepth incomparability for
  transaction-local output tuples:
  \begin{enumerate}[nosep,label=(\alph*)]
    \item \emph{Write skew} ($\mathrm{qdepth}_{\mathrm{SER}} >
          \mathrm{qdepth}_{\mathrm{SI}}$):
          $T_i$ is on an rw-rw cycle.
          SER must break the cycle (making $t$ contingent); SI
          allows it ($t$ is robust).
    \item \emph{FCW-forced abort} ($\mathrm{qdepth}_{\mathrm{SI}} >
          \mathrm{qdepth}_{\mathrm{SER}}$):
          $T_i$ has a write-write conflict with no enclosing cycle.
          SER commits both (ordering them; $t$ is robust); SI aborts
          one via FCW (making $t$ contingent).
  \end{enumerate}
\end{proposition}
\begin{proof}
  (a)~$T_A$ reads $x$, writes $y$; $T_B$ reads $y$, writes $x$
  (write skew).
  Under SI: no write-write conflict, both commit;
  $\mathrm{qdepth}_{\mathrm{SI}}(t) = 0$.
  Under SER: the rw-rw cycle is forbidden; one must abort;
  $\mathrm{qdepth}_{\mathrm{SER}}(t) > 0$.

  (b)~$T_1$ writes $x$ and inserts $t$ into a separate relation $R$;
  $T_2$ writes $x$; no other conflicts.
  Under SER: both commit (no cycle);
  $\mathrm{qdepth}_{\mathrm{SER}}(t) = 0$.
  Under SI: FCW forces one to abort;
  $\mathrm{qdepth}_{\mathrm{SI}}(t) > 0$.
\end{proof}

\begin{corollary}[Compositional isolation insensitivity]
  \label{cor:isolation-migration}
  If every transaction contributing to a provenance query $Q$ is
  isolation-insensitive, then every output tuple of $Q$ is also
  isolation-insensitive: same qdepth under SER and SI.
\end{corollary}
\begin{proof}
  Immediate from Corollary~\ref{cor:propagation-filtration}: positive
  RA cannot increase qdepth beyond its inputs.
\end{proof}

\begin{example}[Compositional insensitivity]
  \label{ex:complement-union}
  Transactions $T_1, T_2, T_3$ each insert into relations
  $R_1, R_2, R_3$.
  None participates in a write-skew or FCW pattern, so each is
  isolation-insensitive.
  A provenance query $Q = R_1 \bowtie R_2 \bowtie R_3$ joins their
  outputs.
  The corollary guarantees that every tuple in $Q$'s result is also
  insensitive---without analyzing $Q$'s interaction with the conflict
  graph.
  This is the compositional extension of Vandevoort et
  al.'s~\cite{vandevoort2025mixed} per-transaction certification:
  the filtration certifies \emph{query results}, not just transactions.
\end{example}

\begin{figure*}[t]
  \centering
  \footnotesize
  \setlength{\tabcolsep}{4pt}
  \renewcommand{\arraystretch}{1.15}

  \begin{tabular}{@{}lll@{}}
    \toprule
    \textbf{Isolation / Case} & \textbf{Resolving determination} & \textbf{Depth} \\
    \midrule
    Read committed &
    (none---no outcome contradicted) & $0$ \\
    SER, no cycles &
    $\{\varphi_{T_1 \prec T_2},\, \varphi_{T_2 \prec T_Q},\, \ldots\}$ (all commute) & $1$ \\
    SER, $d$ visible &
    $\varphi_{T_2 \prec T_Q} \cdot \varphi_{T_Q \prec T_3}$ & $1$ \\
    SER, $d$ absent &
    $\varphi_{T_2 \prec T_3} \cdot \varphi_{T_3 \prec T_Q}$ & $1$ \\
    SER, overlapping cycles &
    $\varphi_{\mathsf{a}}(T_1) \seq \varphi_{\mathsf{a}}(T_2)$
    \;\emph{(Thm.~\ref{thm:txn-depth}c)} & $O(n)$ \\
    Snapshot isolation &
    $\{\varphi_{\snap(T_6){=}\emptyset},\, \varphi_{\snap(T_7){=}\emptyset}\}
      \seq \varphi_{\mathrm{fcw}(S(5),T_6)}
      \seq \varphi_{\snap(T_Q){=}\{\cdots\}}$ & $O(n)$ \\
    \bottomrule
  \end{tabular}

  \smallskip
  {\footnotesize
  $\varphi_{T_i \prec T_j}$ = ordering;\;
  $\varphi_{\mathsf{a}}(T)$ = abort;\;
  $\{\cdots\}$ = commuting set;\;
  $\seq$ = non-commuting sequence.}

  \caption{
    Determination structure across isolation levels
    (Example~\ref{ex:example}), assuming scheduling discretion.
    Depth is $0$ under RC or fully reactive protocols;
    $O(n)$ worst-case under SER or SI with discretion
    (Theorem~\ref{thm:txn-depth}).
  }
  \Description{Table showing resolving determinations and their depth
    across isolation levels under scheduling discretion.}
  \label{fig:det-tables}
\end{figure*}

\section{Beyond Transactions: Negation and the Filtration}
\label{sec:negation-preview}

We now turn to our second instantiation.
In $\mbox{Datalog}^\neg$, the non-monotone operation is
\emph{negation}: concluding $\neg p$ requires establishing that
$p$ has no (further) derivation---a completeness guarantee
that classical approaches handle via monus~\cite{dannert2021semiring}.
Determination provenance decomposes this into a \emph{commitment}
(declaring a set of derivations complete, thereby licensing negation)
followed by monotone evaluation (the entailment that propagates
consequences of that commitment).
The payoff is twofold.
First, determination provenance explains \emph{why-not} under
ambiguity: an atom can be absent because of a semantic commitment, not
a blocked derivation.
Second, the filtration's compositionality lets one check whether a
query result holds under multiple semantics by inspecting supports
algebraically.
(Full details in Appendix~\ref{sec:appendix-datalog}.)

\paragraph{Setup.}
Consider the program
\begin{align*}
  p(c) & \leftarrow a(c).      & r(c) & \leftarrow \neg s(c). \\
  q(c) & \leftarrow \neg p(c). & s(c) & \leftarrow \neg r(c).
\end{align*}
with EDB $\{a(c)\}$.
The stratification depth is $1$: stratum~0 contains $\{a, p\}$, 
stratum~1 contains $\{q\}$ (which negates $p$).
Independently, the pair $(r, s)$ forms an unstratified cycle through negation.

\paragraph{The specification.}
Outcomes assign each atom a value in
$\{\mathbf{t}, \mathbf{f}, \mathbf{u}, \bot\}$
($\bot$ = not yet evaluated, $\mathbf{u}$ = genuinely ambiguous),
ordered $\bot \Ord \mathbf{u} \Ord \{\mathbf{t}, \mathbf{f}\}$
(with $\mathbf{t}$ and $\mathbf{f}$ incomparable); outcomes are total
assignments over the Herbrand base (unevaluated atoms receive $\bot$),
and $o_1 \Ord o_2$ iff $o_1(a) \Ord o_2(a)$ for every atom $a$.
The commitment basis has one \emph{sealing predicate}
$\varphi_{\seal}(p)$ (fixing $p$'s derivation, which
licenses $\neg p$ in $q$'s rule) and two \emph{choice predicates}
$\varphi_{a=v}$ (fixing a $\neg$-cycle atom to
$v \in \{\mathbf{t}, \mathbf{f}\}$).
The specification $\Spec$ admits all assignments consistent with the
program rules and EDB~\cite{gelfond1988stable}.
After all commitments are discharged, two resolved outcomes remain:
$o_r = \{a(c), p(c), r(c)\}$ and $o_s = \{a(c), p(c), s(c)\}$---the
two stable models, each the minimal model of the Gelfond-Lifschitz
reduct under the corresponding truth assignment.
(Our choice predicates enumerate exactly the minimal models of the
local GL reduct for each negative SCC;
Appendix~\ref{sec:appendix-datalog} proves the correspondence.)

The resolving determinations are $\mathcal{D} = \{D^{(r)}, D^{(s)}\}$:
\[
  D^{(r)} = \varphi_{\seal}(p) \seq
    \varphi_{r(c)=\mathbf{t}} \cdot \varphi_{s(c)=\mathbf{f}},
  \qquad
  D^{(s)} = \varphi_{\seal}(p) \seq
    \varphi_{s(c)=\mathbf{t}} \cdot \varphi_{r(c)=\mathbf{f}},
\]
sharing the sealing layer (layer~1, one seal) and branching
at the choice layer (layer~2).
The filtration index counts discharged layers:
$\mathcal{F}_0$ is the state before any commitment;
$\mathcal{F}_1$ reflects the sealing layer
(which all determinations share, so no new distinctions arise); and
$\mathcal{F}_2$ reflects the choice layer where the determinations
diverge.
Concretely:
$\mathcal{F}_0 = \mathcal{F}_1 = \{\emptyset, \mathcal{D}\}$
(all determinations agree through layer~1)
and $\mathcal{F}_2 = 2^{\mathcal{D}}$ (once the choice layer is
discharged, all determinations are distinguished so all subsets are possible).

\[
  \begin{array}{r@{\;=\;}l@{\qquad}r@{\;=\;}l}
    P(p(c)) & \mathcal{D} \times \{x_a\} \;\text{(robust)} &
    P(r(c)) & \{(D^{(r)},\, 1)\} \\
    P(q(c)) & \{\} \;\text{(robustly absent)} &
    P(s(c)) & \{(D^{(s)},\, 1)\}
  \end{array}
\]
Atoms $r(c)$ and $s(c)$ are contingent: each holds under exactly one
determination.
Under $D^{(r)}$, $s(c)$ is absent because of the commitment
$\varphi_{r(c)=\mathbf{t}}$ itself---not a blocked derivation but a
semantic choice.
Reading the filtration only through the sealing layer
($\mathcal{F}_1$ here), these atoms have partial support---they are
\emph{contingent} in our terminology, which is exactly the
well-founded classification $\mathbf{u}$.
Conversely, atoms with full or empty support are \emph{robust}
(WFS's $\mathbf{t}$ and $\mathbf{f}$ respectively): their truth
values are settled without discharging the choice layer.

\begin{theorem}[Negation semantics as filtration levels]
  \label{thm:negation-filtration}
  For a finite $\mbox{Datalog}^\neg$ program whose negative SCCs admit
  a layered choice basis sound and complete for stable models
  (Appendix~\ref{app:datalog-general}), with $k$ stratification
  layers and $d$ negative SCCs in the longest dependency chain, the
  determination semiring has depth $k{+}d$.
  The classical negation semantics correspond to filtration levels:
  \begin{enumerate}[nosep,label=(\alph*)]
    \item Stable-model semantics reads $\mathcal{F}_{k+d}$ (all choice
          layers discharged).
    \item Well-founded semantics reads $\mathcal{F}_k$; atoms with
          full or empty support are \emph{robust} (WFS's $\mathbf{t}$
          and $\mathbf{f}$); atoms with partial support are
          \emph{contingent} (WFS's $\mathbf{u}$).
    \item Stratified semantics reads $\mathcal{F}_k$ and is defined
          only when $d = 0$.
  \end{enumerate}
\end{theorem}
\begin{proof}[Proof sketch]
  The commitment basis has $k$ sealing predicates followed by $d$
  choice layers (one per negative SCC in topological order).
  The sealing prefix is shared by all determinations (stratified
  evaluation is unique), so $\mathcal{F}_k = \{\emptyset, \mathcal{D}\}$.
  Each choice layer $k{+}i$ resolves SCC $C_i$; when SCCs interact,
  later choices depend on earlier outcomes, giving $d$ non-commuting
  layers and depth $k{+}d$.
  WFS reads $\mathcal{F}_k$ (the sealing prefix) and classifies
  residual ambiguity as $\mathbf{u}$.
  Full proof in Appendix~\ref{sec:appendix-datalog}.
\end{proof}

\noindent
The provenance payoff is compositionality: supports compose through
positive RA (Corollary~\ref{cor:propagation-filtration}), so one can
check whether a query result holds under multiple semantics by
inspecting supports algebraically, not re-running evaluation.
Note the structural contrast: negation semantics form a chain within
the filtration, while isolation levels are incomparable resolving
subsets (Observation~\ref{obs:basis-entails-filtration}).

\begin{theorem}[Monus elimination]
  \label{thm:monus-elimination}
  For any finite stratified $\mbox{Datalog}^\neg$ program and any
  naturally ordered, zero-divisor-free, $\omega$-continuous commutative
  semiring $K$:
  the support of every atom computed via sealing commitments over
  $(K, +, \cdot)$ equals the support computed via least fixpoint with
  monus over $(K, \dot{-})$~\cite{dannert2021semiring}.
\end{theorem}

\begin{proof}
  See Appendix~\ref{sec:appendix-datalog}.
\end{proof}

Sealing replaces monus: rather than subtracting provenance
annotations, one commits to a stratum's completeness (sealed-absent
atoms contribute $1_K$; sealed-present atoms contribute $0_K$).
Within each determination, derivational provenance uses ordinary
positive $K$-evaluation.
For unstratified programs, monus is undefined (no unique model to
subtract against), while determination provenance remains well-defined:
each stable model is a determination, and supports track which models
produce each atom.

\section{Consequences of the Filtration}
\label{sec:discussion}

Observation~\ref{obs:basis-entails-filtration} established that
different specifications over the same commitment basis share a single
filtration.
The preceding sections showed that isolation levels
(Section~\ref{sec:transactions}) and negation semantics
(Section~\ref{sec:negation-preview}) are both instances of this
principle.
We now develop two consequences: quantitative specification
comparison and a connection to PDB query complexity.

\subsection{Quantitative Measures}

The filtration assigns each tuple a depth relative to a provenance
query $Q$ (Definition~\ref{def:det-provenance}).
A tuple with $\mathrm{qdepth}(t) = j$ is settled once layers
$1, \ldots, j$ are discharged; layers $j{+}1, \ldots, d$ were
irrelevant.

\begin{proposition}[Compositional depth bound]
  \label{prop:depth-bound}
  If every input tuple of a positive relational provenance query $Q$
  has $\mathrm{qdepth} < k$, then so does every output tuple
  (Cor.~\ref{cor:propagation-filtration}): layers $k, \ldots, d$
  \mbox{are inert.}
\end{proposition}

\noindent
Certifying this is coNP-complete in the residual width
$W_{\geq k}$ (by reduction from robustness,
Theorem~\ref{thm:robustness-conp}), but both instantiations admit
polynomial sufficient checks:
in transactions, absence from $L$-forbidden
cycles~\cite{adya1999weak}; in $\mbox{Datalog}^\neg$,
absence of transitive dependency on any negative
SCC~\cite{gebser2012asp}.

Because different specifications over the same commitment basis share
a filtration (Observation~\ref{obs:basis-entails-filtration}), we can
compare per-tuple qdepth across specifications.
Given a workload executed under specification $I$ (e.g., SER) and a
candidate alternative $I'$ (e.g., SI), the difference decomposes into
two components:

\paragraph{Work regret.}
Tuples that are robust under both $I$ and $I'$ but have
$\mathrm{qdepth}_I(t) > \mathrm{qdepth}_{I'}(t)$ exhibit \emph{work
regret}: the extra commitment layers under $I$ did not change the
outcome.
For transactions this is bidirectional
(Proposition~\ref{prop:ser-si-incomparable}).

\paragraph{Semantic shift.}
Tuples whose support differs between $I$ and $I'$ --- robust under
one, contingent under the other --- represent a genuine change in
meaning, not wasted work.
For transactions, write-skew and FCW patterns
(Proposition~\ref{prop:ser-si-incomparable}) witness the shift in
each direction; for $\mbox{Datalog}^\neg$, shifting from stable-model
to well-founded semantics converts choice-layer atoms to $\mathbf{u}$.

\noindent
Together, work regret and semantic shift measure \emph{how much}
resolution each tuple costs and \emph{which} tuples genuinely change
semantics across specifications.

\subsection{Commitment Responsibility}

Work regret identifies wasted layers; \emph{determination
responsibility} identifies the single commitment most responsible
for a tuple's contingency.
We adapt the Shapley value from cooperative game
theory~\cite{livshits2021shapley}: each commitment in the
determination is a player, and the game measures how much each
commitment contributes to a tuple's contingency.
The Shapley value of commitment $\varphi$ for tuple $t$ is
$\varphi$'s average marginal contribution over all removal orderings
(formal definition in Appendix~\ref{app:responsibility-details}).

In Example~\ref{ex:running-revisited}, responsibility is concentrated
on the ordering commitments that determine whether $T_Q$ observes
$T_2$'s insert before an effective delete.
In deeper filtrations (e.g., the $T_\infty$ construction),
responsibility distributes across layers: each round's victim decision
contributes a fraction of the total contingency, and the Shapley value
quantifies how much.

This lifts the classical tuple-level responsibility of Meliou et
al.~\cite{meliou2010causality} from data tuples to semantic
commitments: where they ask ``which input tuple is most responsible
for an output,'' we ask ``which commitment is most responsible for the
output holding under this resolution.''
Unifying the two into a single joint game is an open question
in Appendix~\ref{app:open-questions}.
Full development, including hardness and approximation results,
appears in Appendix~\ref{app:responsibility-details}.

\subsection{Provenance Query Complexity}

The filtration also connects determination provenance to the
well-studied complexity landscape of probabilistic databases.

\begin{proposition}[Depth-$1$ = tuple-independent PDB]
  \label{prop:layered-pdb}
  A depth-$1$ determination semiring over $n$ binary commitments is
  isomorphic to a tuple-independent PDB over $n$ tuple-existence
  events: supports are positive Boolean formulas, and query evaluation
  has the same complexity.
\end{proposition}

\noindent
At higher depths, the filtration structures the correlations: within
each layer, commitments commute (giving PDB-style lineage), while
cross-layer dependencies introduce conditioning.
Determination provenance thus sits between tuple-independent PDBs
(depth~$1$) and general correlated PDBs (intractable): each layer
admits the PDB complexity toolkit, and depth measures how many layers
must be composed.
Appendix~\ref{app:depth-reduction} develops transformations for
reducing depth.

\section{Related Work}
\label{sec:related}

\noindent\textbf{Algebraic provenance.}
Semiring provenance~\cite{green2007provenance} and its extensions to
probabilistic~\cite{suciu2011probabilistic},
factorized~\cite{olteanu2015factorized}, and
aggregate~\cite{senellart2018provenance} settings operate after the
semantic model is fixed---a single database instance with no
unresolved ambiguity (depth~$0$ in our framework).
Why-not provenance~\cite{buneman2001why}, causal
responsibility~\cite{meliou2010causality}, and counterfactual
explanations~\cite{halpern2001causes} likewise condition on a single
outcome.
Determination provenance generalizes these by making the
commitments that precede derivational provenance explicit and composable.
Determination responsibility (Section~\ref{sec:discussion}) lifts
Shapley-value attribution~\cite{meliou2010causality,livshits2021shapley}
from data tuples to semantic commitments; the two games are sequential,
not competing: determination
responsibility attributes the resolution; tuple-level responsibility
attributes the derivation under that resolution.

\smallskip
\noindent\textbf{Possible-worlds and probabilistic provenance.}
Probabilistic databases~\cite{suciu2011probabilistic} add flat
uncertainty over tuple-existence events (no non-commuting commitment
layers); lineage is a Boolean formula over these events.
A tuple-independent PDB is the depth-$1$ special case of our framework
(Proposition~\ref{prop:collapse-depth-1}): all commitments commute in
a single layer.
The filtration generalizes this to layered, non-commuting commitments;
PDB tractability dichotomies~\cite{dalvi2012probabilistic,suciu2011probabilistic} concern
formula complexity within a single layer---an orthogonal axis that
applies equally within each layer of a determination.

\smallskip
\noindent\textbf{Provenance beyond the monotone fragment.}
Dannert et al.~\cite{dannert2021semiring} extend semiring provenance to
stratified negation within a single model; for stratifiable programs,
our within-determination provenance coincides with theirs.
Deutch et al.~\cite{deutch2014circuits} and
Köhler et al.~\cite{kohler2012declarative} likewise assume a resolved
semantics.
Our contribution addresses unstratified negation, where multiple stable
models coexist and the determination semiring tracks which models
support each tuple.

\smallskip
\noindent\textbf{Provenance in systems settings.}
Database provenance systems (Perm~\cite{glavic2008perm},
GProM~\cite{arab2018gpromprovenance},
Smoke~\cite{psallidas2018smoke}) trace derivations under a fixed
execution.
Vandevoort et al.~\cite{vandevoort2025mixed} study which transactions can
run at weaker isolation while preserving serializable behavior---a
cross-specification support comparison in our framework.
Our framework separates semantic commitments from derivational explanation,
enabling per-execution provenance to be composed across alternative
resolutions (Appendix~\ref{app:systems-lineage}).

\smallskip
\noindent\textbf{Other related frameworks.}
Ameloot et al.~\cite{ameloot2013relational} characterize when
coordination-free evaluation is possible;
Hellerstein~\cite{hellerstein2025complexity} develops a related complexity theory for
the cost of semantic resolution.
We take determinations as given and explain how query results depend
on them algebraically.
Arenas et al.~\cite{arenas1999consistent} and
Wijsen~\cite{wijsen2019certain} define certain answers as
tuples holding in every repair---analogous to robustness,
but over flat repairs rather than layered determinations.

\section{Conclusion}
\label{sec:conclusion}

Determination provenance makes the commitments that resolve semantic
ambiguity explicit and composable.
Supports over resolving determinations form a filtered semiring;
the filtration measures query-relative depth, is respected by positive
relational algebra, and recovers classical provenance at depth~$0$.
Transactions and $\mbox{Datalog}^\neg$ instantiate the framework;
in both, classical semantic variants (isolation levels; negation
semantics) correspond to different views of a shared filtration.
The filtration enables quantitative diagnosis --- work regret,
semantic shift, and per-commitment responsibility --- that goes beyond
binary robustness to explain \emph{how much} and \emph{why} a tuple
depends on each layer of resolution.

This work extends algebraic data
provenance to settings like transactional databases and
distributed systems, where semantic ambiguity is the norm and where
provenance can pay off in full.
The practical implication is a recipe: given a system's event trace,
define an outcome space and refinement order that identifies the
commitments.
The full algebraic machinery then applies---robustness,
counterfactuals, responsibility attribution, cross-specification
comparison---over existing observability infrastructure.
Appendix~\ref{app:systems-lineage} develops this direction.

Looking further ahead, the determination semiring has the structure of
a probability space: $\mathcal{D}$ is a sample space, supports are
measurable sets, and the filtration---which emerged here from the
algebraic structure of commitments---is a filtration in the
measure-theoretic sense (level-$k$ supports are exactly the
$\mathcal{F}_k$-measurable sets).
Under a probability measure on $\mathcal{D}$, support ratios become
probabilities, work regret becomes expected regret, and responsibility
becomes conditional expectation.
That the algebraic structure of determination provenance lands so
close to the machinery of stochastic processes suggests a deeper
connection; developing it is a natural next step.

The remaining appendices provide:
algebraic details (Appendix~\ref{app:algebra}),
robustness and responsibility proofs
(Appendices~\ref{app:robustness-proofs}
and~\ref{app:responsibility-details}),
worked examples for transactions
(Appendices~\ref{sec:appendix-concurrency}
and~\ref{app:protocols}) and $\mbox{Datalog}^\neg$
(Appendix~\ref{sec:appendix-datalog}),
heredity canonicalization
(Appendix~\ref{app:persistence-canonicalization}),
depth reduction (Appendix~\ref{app:depth-reduction}),
and open theory directions
(Appendix~\ref{app:open-questions}).

\pagebreak
\appendix

\section{Algebraic Details}
\label{app:algebra}
\label{app:det-semiring-details}

This appendix develops the algebraic structures introduced in
Section~\ref{sec:algebra}: first the structure \emph{within} a single
determination, then the determination semiring and filtration that
operate \emph{across} determinations.

\subsection{Within a Determination}

A single determination of depth~$k$ is a sequence of conditional
commutative monoids connected by the sequencing operator $\seq$:
\[
  (M_1, \circ_1) \;\seq\; (M_2, \circ_2) \;\seq\; \cdots \;\seq\; (M_k, \circ_k)
  \;\seq\; (K, +, \cdot).
\]
Each $M_\ell$ is the commutative monoid of commitments at layer~$\ell$:
within a layer, commitments commute
($\Spec^{\ell}(H \cdot \varphi \cdot \psi) = \Spec^{\ell}(H \cdot \psi \cdot \varphi)$);
across layers they do not, because later commitments depend on which
outcomes earlier commitments eliminated.
The operator $\seq$ discharges one layer's monoid and conditions the
specification for the next; it is not an algebraic operation.

The terminal element is the derivational semiring $K$ of Green et
al.~\cite{green2007provenance}, which requires a \emph{fixed} database
instance.
Prior to resolution, $\mathsf{obs}$ maps multiple admissible outcomes
to different instances; after resolution, exactly one outcome remains
and $K$-relational evaluation applies.
Classical provenance is the degenerate case $k = 0$: no commitment
layers, just $K$.

\subsection{Support Propagation Under Relational Algebra}

The set $2^{\mathcal{D}}$ of all supports forms a commutative semiring
under $(\cup, \cap, \emptyset, \mathcal{D})$: union is addition,
intersection is multiplication, and distributivity holds because
$(2^{\mathcal{D}}, \cup, \cap)$ is a distributive lattice.
The following proposition shows that positive RA evaluation respects
this structure.

\begin{proposition}
  The supports $\mathrm{supp}(P(t)) \subseteq \mathcal{D}$ compose
  through positive relational algebra as follows
  (assuming $K$ is zero-divisor-free, as holds for $\mathbb{N}[X]$
  and all standard provenance semirings):
  \begin{enumerate}[nosep,label=(\alph*)]
    \item \emph{Join:}
          $\mathrm{supp}(P(t_1 \bowtie t_2))
          = \mathrm{supp}(P(t_1)) \cap \mathrm{supp}(P(t_2))$.
    \item \emph{Union:}
          $\mathrm{supp}(P(t_1) \cup P(t_2))
          = \mathrm{supp}(P(t_1)) \cup \mathrm{supp}(P(t_2))$.
    \item \emph{Projection/Selection:}
          $\mathrm{supp}(P(\pi(t))) \supseteq \mathrm{supp}(P(t))$
          and
          $\mathrm{supp}(P(\sigma(t))) \subseteq \mathrm{supp}(P(t))$.
  \end{enumerate}
\end{proposition}
\begin{proof}
  (a)~A join requires both premises to hold simultaneously.
  Under determination $D$, $t_1 \bowtie t_2$ is nonempty iff both
  $P_D(t_1) \neq 0$ and $P_D(t_2) \neq 0$ (semiring multiplication
  has no zero divisors in $\mathbb{N}[X]$).
  Hence $D \in \mathrm{supp}(P(t_1 \bowtie t_2))$ iff
  $D \in \mathrm{supp}(P(t_1)) \cap \mathrm{supp}(P(t_2))$.

  (b)~A union provides alternative derivation paths.
  $P_D(t_1 \cup t_2) = P_D(t_1) + P_D(t_2)$; this is nonzero iff at
  least one summand is nonzero.

  (c)~Projection can only add derivation paths (combining tuples that
  agree on projected attributes); selection can only remove them.
\end{proof}

\subsection{Filtration Details}

The filtration $\mathcal{F}_0 \subseteq \mathcal{F}_1 \subseteq \cdots
\subseteq \mathcal{F}_d = 2^{\mathcal{D}}$ is defined in
Section~\ref{sec:filtration}.
Here we verify that it interacts correctly with the determination
semiring operations and with query evaluation.

\begin{proposition}[Filtration closure under semiring operations]
  For any $k$, if $S_1, S_2 \in \mathcal{F}_k$, then
  $S_1 \cup S_2 \in \mathcal{F}_k$ and
  $S_1 \cap S_2 \in \mathcal{F}_k$.
\end{proposition}
\begin{proof}
  A set is in $\mathcal{F}_k$ iff it is a union of $\equiv_k$ classes
  (Definition~\ref{def:level-k-agreement}).
  The union of two such sets is again a union of $\equiv_k$ classes.
  For intersection: if $D \in S_1 \cap S_2$ and $D' \equiv_k D$,
  then $D' \in S_1$ (since $S_1 \in \mathcal{F}_k$) and
  $D' \in S_2$ (since $S_2 \in \mathcal{F}_k$), so
  $D' \in S_1 \cap S_2$.
\end{proof}

\begin{corollary}[Query evaluation preserves filtration level]
  If $\mathrm{supp}(P(t_1)) \in \mathcal{F}_j$ and
  $\mathrm{supp}(P(t_2)) \in \mathcal{F}_k$, then
  $\mathrm{supp}(P(t_1 \bowtie t_2)) \in \mathcal{F}_{\max(j,k)}$ and
  $\mathrm{supp}(P(t_1 \cup t_2)) \in \mathcal{F}_{\max(j,k)}$.
  Positive RA cannot increase query-relative depth beyond the
  maximum depth of its inputs.
\end{corollary}

\section{Robustness Proofs}
\label{app:robustness-proofs}

\begin{theorem}[Robustness is coNP-complete in width]
  \label{thm:robustness-conp}
  For any instantiation where the widest layer has $n$ independent
  binary commitments, validity checking is polynomial, and query
  evaluation is polynomial, deciding robustness is coNP-complete.
\end{theorem}

\subsection{Transactional Hardness (Full Proof)}

We reduce from \textsc{DNF-Validity}: given a DNF formula
$\psi = t_1 \vee \cdots \vee t_\ell$ over variables $x_1, \ldots, x_n$,
is $\psi$ true under every truth assignment?
This problem is coNP-complete.

\paragraph{Construction.}
Using the ordering basis (Section~\ref{sec:scheduling-basis}) with all
transactions committing:
for each variable $x_i$, create two transactions $T_i^+$ and $T_i^-$,
each writing to a distinct object $o_i$.
$T_i^+$ writes $\mathsf{true}$; $T_i^-$ writes $\mathsf{false}$.
Both are concurrent; their write-write conflict on $o_i$ creates an edge.
The ordering commitment $\varphi_{T_i^+ \prec T_i^-}$ or
$\varphi_{T_i^- \prec T_i^+}$ determines the final value of $o_i$
(last-writer-wins).
Since each pair conflicts only on its own object, no cycles arise and
every combination of orderings is a valid resolving determination.
The set $\mathcal{D}$ has $2^n$ elements, one per truth assignment.

\paragraph{Query.}
For each term $t_j = \ell_1 \wedge \cdots \wedge \ell_m$ of $\psi$,
define a conjunctive query $Q_j$ that selects a distinguished tuple
$t^*$ iff the objects corresponding to the literals of $t_j$ have the
matching values.
The overall query is $Q = Q_1 \cup \cdots \cup Q_\ell$ (a union of
conjunctive queries).

\paragraph{Correctness.}
Under any resolving determination $D$, the ordering of each pair
$(T_i^+, T_i^-)$ induces a truth assignment $\alpha_D$ to $x_i$.
The tuple $t^*$ is present in $Q$'s result under $D$ iff $\alpha_D$
satisfies at least one term of $\psi$---that is, iff
$\alpha_D \models \psi$.
Hence $t^*$ is robust (present under every $D \in \mathcal{D}$) iff
$\psi$ is valid (true under every assignment).

\paragraph{Polynomial conditions.}
Condition~(i): each determination is a sequence of $n$ ordering
commitments (polynomial-size).
Condition~(ii): checking validity amounts to verifying that the
committed conflict graph under the chosen ordering has no forbidden
cycles (trivially true here, since no cycles exist).
Condition~(iii): evaluating a UCQ over a polynomial-size database is in
polynomial time.

\subsection{Datalog Hardness}

The general theorem (Theorem~\ref{thm:robustness-conp}) applies directly
to the $\mbox{Datalog}^\neg$ instantiation, providing an independent proof of
coNP-hardness via the DNF-Validity reduction.

\paragraph{Instantiation.}
Given $n$ disjoint $\neg$-cycles of length~2, the choice layer has $n$
independent binary commitments (one per cycle).
The DNF-Validity gadget of the transactional proof applies verbatim:
encode a DNF $\psi$ as a UCQ over the choice-induced truth values, so
that the distinguished tuple is robust iff $\psi$ is valid.
This yields coNP-hardness for the $\mbox{Datalog}^\neg$ instantiation.

The result coincides with the known coNP-completeness of skeptical
reasoning under stable-model
semantics~\cite{marek1991autoepistemic}---the robustness question
``is $\mathrm{supp}(P(a)) = \mathcal{D}$?'' is exactly the skeptical
reasoning question ``does $a$ hold in every stable model?''
The match validates that the determination semiring's width-based
abstraction captures the right source of hardness.

\paragraph{Polynomial conditions.}
Condition~(i): a determination is a sequence of sealing and choice
commitments, one per stratum plus one per $\neg$-cycle atom
(polynomial-size).
Condition~(ii): checking that a determination produces a valid stable
model requires computing the Gelfond--Lifschitz
reduct~\cite{gelfond1988stable} and verifying that the proposed model
is its minimal model (polynomial time).
Condition~(iii): evaluating a $\mbox{Datalog}^\neg$ query over a fixed model is in
polynomial time.

\section{Responsibility: Additional Results}
\label{app:multi-layer-resp}
\label{app:responsibility-details}
\label{sec:responsibility}

This appendix develops determination responsibility: a Shapley-value
measure of each commitment's contribution to a contingent tuple.

Fix $n$ independent binary commitments
$\varphi_1, \ldots, \varphi_n$, a realized determination $D^*$
(the resolution that actually occurred), and a
tuple $t$ with support $S \subseteq \mathcal{D}$.
The \emph{presence game} is the cooperative game $(N, v)$ with player
set $N = \{1,\ldots,n\}$: for a coalition $C \subseteq N$,
\[
v(C) \;\triangleq\;
\frac{|\{D \in S \mid D \text{ agrees with } D^* \text{ on every }
\varphi_i,\; i \in C\}|}{2^{n-|C|}}.
\]
Intuitively, $v(C)$ is the probability that $t$ holds when the
outcomes of the commitments in $C$ have been revealed (matching $D^*$)
and the remaining commitments are still unknown (uniformly random).
The \emph{determination responsibility} of $\varphi_i$ is its Shapley
value in this game.

\begin{theorem}[Hardness]
  Computing determination responsibility is $\mathrm{\#P}$-hard, even
  for single-layer determinations with support given as DNF.
\end{theorem}
\begin{proof}[Proof sketch]
  We follow the approach of~\cite{livshits2021shapley}.
  Reduction from counting satisfying assignments of a monotone DNF.
  Given a DNF $\phi$ over variables $x_1,\ldots,x_n$, construct a
  determination with one binary commitment per variable and support
  $S = \mathrm{sat}(\phi)$.
  Then $v(\{i\}) - v(\emptyset)$ requires computing
  $|\mathrm{sat}(\phi|_{x_i = D^*_i})|$, which is $\mathrm{\#P}$-hard
  for monotone DNF~\cite{livshits2021shapley}.
\end{proof}

\begin{theorem}[Tractability at bounded support treewidth]
  \label{thm:responsibility-tw}
  For a single-layer determination whose monotone support formula has
  primal-graph treewidth $w$, determination responsibility is computable
  in time $O(2^w \cdot n^2 \cdot p(n))$---FPT in $w$.
\end{theorem}
\begin{proof}[Proof sketch]
  We follow the approach of~\cite{livshits2021shapley}.
  Computing $v(C)$ reduces to weighted model counting on the support
  formula with variables in $C$ fixed.
  On bounded-treewidth formulas, weighted model counting runs in
  $O(2^w \cdot p(n))$ via dynamic programming on a tree
  decomposition~\cite{livshits2021shapley}.
  The Shapley value requires $O(n^2)$ such evaluations (one per
  coalition size and player), giving the stated bound.
\end{proof}

\subsection{Additive Approximation}

When support treewidth is unbounded, sampling provides an efficient
additive approximation.

\begin{proposition}[Additive approximation]
  \label{prop:fpras}
  For any $\varepsilon, \delta > 0$, $\rho(\varphi_i, t)$ can be
  additively $\varepsilon$-approximated with probability
  $\ge 1 - \delta$ in time
  $O(n/\varepsilon^2 \cdot \log(1/\delta) \cdot p(n))$.
\end{proposition}
\begin{proof}
  By standard permutation sampling of Shapley
  values~\cite{livshits2021shapley,shapley1953value}: each sample draws
  a random permutation, fixes preceding players to $D^*$-values,
  completes the rest uniformly, and measures $\varphi_i$'s marginal
  effect.
  Hoeffding's inequality gives the bound.
\end{proof}

\subsection{Worked Example: Multi-Layer Responsibility}

For depth $d > 1$, responsibility is defined per layer: the
responsibility of a commitment at layer $k$ is its Shapley value in the
presence game conditioned on layers $1, \ldots, k{-}1$ being discharged.
This respects non-commutativity---one does not ask ``what if layer-$2$
commitments were resolved before layer-$1$?'' because that is not a
valid determination.
Complexity per layer is the same as the single-layer case.
The filtration provides the stopping criterion: if
$\mathrm{qdepth}(t) = k$, only layers $1, \ldots, k$ contribute
nonzero responsibility.

\begin{example}[Multi-layer responsibility]
  \label{ex:multi-layer-resp}
  Three transactions with two overlapping directed 2-cycles:
  $T_1 \rightleftharpoons T_2$ and $T_2 \rightleftharpoons T_3$.
  Under the ordering basis, each 2-cycle requires exactly one abort.
  Let $\varphi_1$ resolve the first cycle (abort $T_1$ or $T_2$) and
  $\varphi_2$ resolve the second (abort $T_2$ or $T_3$).
  If $\varphi_1$ aborts $T_2$, the second cycle is auto-resolved
  (depth~$1$); otherwise both commitments are needed (depth~$2$).

  Tuple $t$ is written by $T_3$; it holds iff $T_3$ is not aborted.
  The determination space is a tree (not a product): if $\varphi_1$
  aborts $T_2$, layer~2 is never discharged.
  \begin{center}\small
  \begin{tabular}{lll}
    & commitments & outcome \\
    \hline
    $D_1$: & abort $T_2$ & $t$ holds \\
    $D_2$: & abort $T_1$,\; abort $T_2$ & $t$ holds \\
    $D_3$: & abort $T_1$,\; abort $T_3$ & $t$ absent
  \end{tabular}
  \end{center}
  Support: $|\mathrm{supp}(t)| = 2$ out of $3$.
  Prior $\Pr[t \text{ holds}] = 2/3$ (uniform over determinations).

  Responsibility at each layer is the change in $t$'s conditional
  presence probability caused by that layer's realized commitment:
  $\rho(\varphi_k) = \Pr[t \text{ holds} \mid \text{layers } 1{\ldots}k
  \text{ fixed to } D^*] - \Pr[t \text{ holds} \mid \text{layers }
  1{\ldots}(k{-}1) \text{ fixed}]$.
  Total responsibility sums across layers to $1 - \Pr[t \text{ holds}]$
  when $t$ holds in $D^*$ (the gain from prior uncertainty to certainty).

  For $D^* = D_2$ --- we observe $t$ holds, with both $T_1$ and $T_3$
  surviving (aborts: $T_1$ at layer~1, $T_2$ at layer~2):
  $\rho(\varphi_1) = 1/2 - 2/3 = -1/6$ (negative---keeping
  $T_2$ alive reduced $t$'s probability from $2/3$ to $1/2$);
  $\rho(\varphi_2) = 1 - 1/2 = 1/2$ (positive---directly
  saving $T_3$ raised it from $1/2$ to certainty).
  Total: $-1/6 + 1/2 = 1/3 = 1 - 2/3$.

  For $D^* = D_1$ --- we observe $t$ holds, with $T_2$ aborted
  (resolving both cycles at layer~1):
  $\rho(\varphi_1) = 1 - 2/3 = 1/3$, no layer-2 responsibility.
  The same total budget ($1/3$) is concentrated entirely on the single
  commitment that made $t$ certain.

  For $D^* = D_3$ --- we observe $t$ absent ($T_3$ aborted):
  $\rho(\varphi_1) = 1/2 - 2/3 = -1/6$;
  $\rho(\varphi_2) = 0 - 1/2 = -1/2$.
  Total: $-1/6 + (-1/2) = -2/3 = 0 - 2/3$.
  Both commitments bear negative responsibility for $t$'s absence.
\end{example}

\section{Transactional Determination Provenance: SI and Cross-Level Comparison}
\label{sec:appendix-concurrency}

The body develops determination provenance for serializability in
detail (Examples~\ref{ex:example} and~\ref{ex:running-revisited}).
This appendix extends the treatment to snapshot isolation (SI), where
first-committer-wins (FCW) introduces a different commitment
structure, and then characterizes the fine-grained separation between
SER and SI in terms of query-relative depth.

\subsection{Snapshot Isolation}
\label{sec:appendix-SI}

\paragraph{Fully reactive SI: depth $0$.}
Under a fully reactive SI implementation, snapshot assignment is
deterministic (each transaction reads the committed state at its start
time) and first-committer-wins (FCW) is deterministic (whichever
transaction's commit request arrives first wins on a contested object;
the other aborts).
All decisions are entailed by the environment; depth is~$0$.

\paragraph{SI with scheduling discretion: depth $O(n)$.}
When the system has discretion over commit ordering (choosing which
concurrent writer commits first) or batching (choosing when to assign
snapshots), depth grows to $O(n)$ via the same structural argument as
serializability.

\begin{proposition}[Worst-case SI depth]
  \label{prop:si-depth-n}
  Under snapshot isolation with commit-order discretion, worst-case
  determination depth is $\Theta(n)$.
\end{proposition}
\begin{proof}
  The construction mirrors Theorem~\ref{thm:txn-depth}.
  Consider $n$ pairs of transactions $(T_1, T_1'), (T_2, T_2'), \ldots,
  (T_n, T_n')$, where $T_i$ and $T_i'$ both write object $x_i$
  (creating an FCW conflict), and $T_{i+1}$ reads $x_i$ before writing
  $x_{i+1}$ (so its snapshot depends on who won FCW on $x_i$).

  At each layer $i$, the system chooses the FCW winner on $x_i$.
  This determines which version of $x_i$ appears in $T_{i+1}$'s
  snapshot, potentially affecting $T_{i+1}$'s write to $x_{i+1}$
  (e.g., $T_{i+1}$ writes $x_{i+1}$ only if it reads a specific value
  of $x_i$).
  The FCW decision at layer $i+1$ therefore depends on the outcome at
  layer $i$: the two decisions do not commute.

  This chain of FCW-to-snapshot dependencies gives depth $n$.
  The upper bound ($n$ transactions, each resolved by at most one
  commitment) gives $O(n)$.
\end{proof}

\paragraph{Per-batch depth.}
Within a single batch of pending requests, depth is~$2$
(seal + commit-order selection), as in serializability
(Theorem~\ref{thm:txn-depth}).

\paragraph{Worked example.}
To illustrate FCW-induced depth concretely, consider two transactions
writing the same object concurrently.
Assume the database initially contains $S(5,e)$, and:
\[
  \begin{array}{ll}
    T_6: \textsf{write } S(5,e') \qquad &
    T_7: \textsf{write } S(5,e'').
  \end{array}
\]
A query transaction $T_Q$ reads the value of key $5$ in $S$ after both
$T_6$ and $T_7$ request commit.
Under SI, only one of $T_6$ or $T_7$ can commit (FCW on key $5$);
the other must abort.
If the system has commit-order discretion, this is a genuine
commitment: choosing $T_6$ as the FCW winner determines that $T_Q$
reads $e'$; choosing $T_7$ gives $e''$.
The determination has depth~$2$ (seal the batch of commit requests,
then choose the winner):
\[
  D_{\mathrm{SI}} =
    \varphi_{\mathsf{seal}} \seq \varphi_{\mathrm{fcw}(S(5),\,T_6)}.
\]
Under $D_{\mathrm{SI}}$, $T_Q$ reads $S(5,e')$; the derivational
provenance is $P_{D_{\mathrm{SI}}}(e') = x_6$.
Under the alternative $\varphi_{\mathrm{fcw}(S(5),T_7)}$, $T_Q$ reads
$S(5,e'')$---a contingent outcome.

\subsection{Fine-Grained Separation of SER and SI}
\label{app:ser-si-separation}

Proposition~\ref{prop:ser-si-incomparable} establishes that
query-relative depth is incomparable across SER and SI.
Here we develop the characterization in full.

\paragraph{Setup.}
Fix a workload $W$ (a set of transactions with their read/write sets)
and a conflict graph $G$ induced by $W$.
For each isolation level $L \in \{\mathrm{SER}, \mathrm{SI}\}$, the
determination semiring $\mathcal{D}_L$ and filtration
$\mathcal{F}^L_0 \subseteq \cdots \subseteq \mathcal{F}^L_d$ are
defined over the same workload but with different admissibility
constraints.
A tuple $t$ derived from transaction $T_i$ has
$\mathrm{qdepth}_L(t) = 0$ iff $T_i$ commits in every resolving
determination under $L$ (robust), and $\mathrm{qdepth}_L(t) > 0$ iff
$T_i$'s fate varies across determinations (contingent).

\paragraph{Characterization.}
The qdepth of $t$ (derived from $T_i$) differs across levels precisely
when $T_i$ participates in a conflict pattern that one level must
resolve but the other does not:

\begin{enumerate}[label=(\alph*),nosep]
  \item \emph{Write skew} ($\mathrm{qdepth}_{\mathrm{SER}}(t) >
        \mathrm{qdepth}_{\mathrm{SI}}(t)$):
        $T_i$ is on a cycle whose edges are all rw-type (anti-dependencies).
        SER forbids all cycles: at least one transaction on the cycle
        must abort, making $t$ contingent ($\mathrm{qdepth}_{\mathrm{SER}} > 0$).
        SI allows rw-only cycles: all transactions commit in every SI
        determination, so $t$ is robust ($\mathrm{qdepth}_{\mathrm{SI}} = 0$).

  \item \emph{FCW-forced abort} ($\mathrm{qdepth}_{\mathrm{SI}}(t) >
        \mathrm{qdepth}_{\mathrm{SER}}(t)$):
        $T_i$ has a write-write conflict with $T_j$ on some object $x$,
        but $T_i$ and $T_j$ are not on any cycle in $G$ (no
        serialization anomaly).
        Under SER: both commit (ordered on $x$, no cycle results);
        $t$ is robust ($\mathrm{qdepth}_{\mathrm{SER}} = 0$).
        Under SI: FCW on $x$ forces one to abort; $t$ holds only if
        $T_i$ wins; $\mathrm{qdepth}_{\mathrm{SI}} > 0$.
\end{enumerate}

\paragraph{Tight condition for qdepth equality.}
$\mathrm{qdepth}_{\mathrm{SER}}(t) = \mathrm{qdepth}_{\mathrm{SI}}(t)$
iff $T_i$ is not involved in either pattern above: it participates in
no rw-only cycle (so SER and SI agree on whether the cycle must be
broken) and in no write-write conflict without an enclosing cycle (so
SER and SI agree on whether both writers commit).
These are exactly the transactions whose commit/abort fate is the same
under both levels---the ``robustly portable'' transactions
of~\cite{vandevoort2025mixed}.

\paragraph{Compositionality.}
By Corollary~\ref{cor:propagation-filtration}, positive relational
algebra cannot increase qdepth.
Hence if all base tuples in a query have equal qdepth across SER and
SI, the query result does too.
A query result can have different qdepth across levels only if at least
one of its base tuples participates in a write-skew or FCW pattern.
This gives a sufficient condition for safe isolation-level change:
if no base tuple of a query is in the SER/SI gap, the query's
provenance is identical under both levels.

\subsection{Certificates and Semantic-Change Queries}

A \emph{certificate} over-approximates a commitment's effect: it
summarizes which outcomes a commitment excludes, admitting possibly more
but never fewer (formal definition in
Appendix~\ref{app:systems-lineage}).
Under serializability, a certificate need only record ordering predicates
sufficient to justify the observation.
Under snapshot isolation, answering semantic-change queries---such as
whether an SI outcome admits a serial explanation---requires preserving
the snapshot structure itself.
Certificates are query-parametric: richer
questions require retaining more semantic information
(Appendix~\ref{app:systems-lineage}).

\section{Protocol-Specific Determination Structures}
\label{app:protocols}

Theorem~\ref{thm:txn-depth} establishes protocol-agnostic depth
bounds via the $T_\infty$ witness.
Here we show how 2PL and MVTO instantiate the ordering basis and
confirm that the bounds are tight.
OCC requires no separate treatment: its determination structure (seal +
transaction-ordering layer) is the body's default presentation
(Example~\ref{ex:overlapping-cycles}), with all conflict detection and
aborts entailed by the chosen validation order.
We first compare per-batch structure across all three protocols, then
give protocol-specific details for 2PL and MVTO.

\subsection{Per-Batch Structure}

\paragraph{2PL with batching.}
Under 2PL, lock acquisition and release are fully reactive (entailed
by the arrival order of requests); the only discretion is
\emph{victim selection} when a deadlock cycle forms.
Without sealing, victim choices cascade: aborting a victim releases
locks, enabling new requests that may form new cycles, giving
$\Theta(n)$ depth.
Sealing a batch (refusing new lock requests) cuts this feedback loop:
all deadlock cycles in the sealed wait-for graph are present
simultaneously, and victim choices across distinct cycles commute
(aborting one victim cannot create a new cycle when no new edges can
form).
Per-batch depth: $2$ (seal + one commuting layer of victim selections).

\paragraph{OCC with batching.}
A batch of pending commit requests is sealed; the system then chooses
a validation order.
Transactions validated earlier can commit (if no cycle among
already-committed); those validated later may find cycles and abort.
The commit-order choice is isomorphic to 2PL victim selection: the
last transaction validated on a cycle is the one that aborts.
Per-batch depth: $2$ (seal + ordering layer).

\paragraph{MVTO with batching.}
For each batch, the system starts with existing transactions (with
timestamps) and an acyclic conflict graph whose topological sort is
consistent with those timestamps.
The batch contains new requests (starts needing timestamp assignment,
reads, writes, commit requests).
The system's only choice is the order in which to process the batch;
all responses (timestamp assignments, read values, forced aborts for
transactions that read between writes, commit acknowledgments) are
entailed by the chosen order.
Per-batch depth: $2$ (seal + action-ordering layer).

\subsection{Two-Phase Locking (2PL)}

Under strict 2PL, a transaction acquires locks before accessing objects
and holds them until commit.
The following details supplement the protocol-agnostic analysis above.

\paragraph{Lock acquisition as ordering commitment.}
Each lock acquisition implicitly fixes a serialization ordering:
if $T_i$ holds a write lock on $x$ and $T_j$ requests a
conflicting lock, $T_j$ blocks until $T_i$ commits (strict 2PL holds
locks until commit), irrevocably establishing $T_i \prec T_j$.
Under the ordering basis
(Section~\ref{sec:scheduling-basis}), this is exactly the ordering
commitment $\varphi_{T_i \prec T_j}$.
The key difference from the protocol-agnostic view is \emph{distribution}:
ordering commitments are made incrementally as locks are acquired
(throughout execution), rather than concentrated at commit time.

\paragraph{Deadlock as forced abort.}
A deadlock cycle $T_1 \to T_2 \to \cdots \to T_m \to T_1$ in the
wait-for graph means the lock-acquisition commitments already made are
jointly inconsistent with all transactions committing.
The system resolves this by aborting one transaction (the ``victim''),
which is an abort commitment $\varphi_{\mathsf{abort}(T_v)}$.
The choice of victim is itself a commitment: different victims yield
different determinations.

\paragraph{Determination structure.}
Under strict 2PL (locks held until commit), ordering commitments
made during the growing phase on distinct objects commute (when
read/write sets are static, ordering $T_i$ on object $x$ does not
affect the lock state of unrelated object $y$).
Non-commutativity arises when deadlock forces a victim choice that
affects subsequent lock availability.
In the fully reactive case (deadlock detected and resolved immediately
on each lock request, with the requestor as victim), all decisions are
entailed and depth is~$0$.
With victim-selection discretion, depth grows to $O(n)$
(Theorem~\ref{thm:txn-depth}): deadlock cycles in the wait-for
graph correspond to conflict cycles requiring resolution.

\subsection{Multiversion Timestamp Ordering (MVTO)}

Under MVTO~\cite{reed1978naming}, each transaction $T_i$ receives a
unique timestamp $ts_i$ at start.
Reads and writes are governed by the timestamp order: $T_i$ reads the
latest version of $x$ written by a transaction with timestamp
$\le ts_i$, and $T_i$'s write of $x$ creates a new version tagged
$ts_i$.
A write by $T_i$ is rejected (and $T_i$ aborted) if some transaction
$T_j$ with $ts_j > ts_i$ has already read an earlier version of $x$:
$T_j$'s read arrived before $T_i$'s write, so the system already
committed to $T_j$ seeing the old version, and $T_i$'s write would
retroactively invalidate that commitment.

\paragraph{Semantic commitments.}
In the fully reactive case (timestamps assigned on arrival), the
primary commitment is timestamp assignment: assigning $ts_i < ts_j$
irrevocably fixes $T_i \prec T_j$ in the serial order on every object
they both access.
This is an ordering commitment $\varphi_{T_i \prec T_j}$ in the
ordering basis (Section~\ref{sec:scheduling-basis}), with a structural
property: a single timestamp assignment simultaneously fixes the
ordering of $T_i$ relative to \emph{all} other active transactions.
With batching, the commitment is finer-grained: the system chooses a
processing order over individual \emph{actions} (reads, writes, starts),
not just transactions.
This action-level ordering subsumes timestamp assignment (timestamps
are assigned when a transaction's start or first read/write operation is
processed) and determines which reads see which versions, hence which
transactions are forced to abort.

\paragraph{Forced aborts.}
Unlike 2PL, MVTO never blocks: conflicting operations are either
served from an appropriate version (reads) or rejected immediately
(writes that arrive too late---after a higher-timestamped transaction
has already read an earlier version).
Rejection forces an abort, but this is not a free choice---it
is entailed by the timestamp assignments already made.
MVTO is therefore deadlock-free.

\paragraph{Determination structure.}
The determination under MVTO mirrors the other protocols (seal +
ordering layer) but at action granularity rather than transaction
granularity.
The action ordering is quotiented by commutativity of adjacent
non-conflicting actions (swapping two actions on disjoint objects has
no observable effect), and forced aborts add no independent
commitments---they are entailed by the ordering.
Fully reactive MVTO has depth~$0$ (timestamp assignment is
entailed by environment event arrival order, and aborts are forced).
With batching, per-batch depth is~$2$: the seal commitment (declaring
the batch complete) followed by a single commitment choosing the
processing order over buffered actions (one total order determines all
version visibility and forced aborts).

\subsection{Comparison Across Protocols.}
All three protocols share a common determination structure (seal +
resolution layer) and the same worst-case depth bounds
(Theorem~\ref{thm:txn-depth}): depth~$0$ when fully reactive,
$\Theta(n)$ with discretion, depth~$2$ per batch.
They differ in commitment basis:
OCC uses a transaction-ordering basis (choose a validation order);
MVTO uses an action-ordering basis (choose a processing order over starts,
reads and writes);
2PL uses a victim-selection basis (choose which transaction to abort
on deadlock).
These bases are incomparable---they record different aspects of the
system's discretion---but yield the same depth structure.

\section[Datalog with Negation: Worked Determination Provenance]{Datalog$^{\neg}$ Worked Determination Provenance}
\label{sec:appendix-datalog}

This appendix instantiates determination provenance for
$\mbox{Datalog}^\neg$ in full generality, then illustrates the
machinery on a running example.
We first develop the canonical commitment basis for programs with
arbitrarily nested negative cycles, prove soundness and completeness,
and state the general semantics-as-filtration theorem.
We then show how the three classical negation semantics (stable,
well-founded, stratified) are recovered as different readings of the
determination semiring, and close with the monus elimination proof.

\begin{example}[Running example]
  \label{ex:datalog-negation}
  Consider the following program $P$ with EDB fact $a(c)$:
  \begin{align*}
    p(c) & \leftarrow a(c).      & r(c) & \leftarrow \neg s(c). \\[1mm]
    q(c) & \leftarrow \neg p(c). & s(c) & \leftarrow \neg r(c).
  \end{align*}
  The pair $(P, \{a(c)\})$ has a stratified fragment ($p, q$: one
  stratum boundary) and a single negative SCC ($r \leftrightarrow
  \neg s$).
  Parameters: $k = 1$, $d = 1$.
  We instantiate the general constructions on this example throughout.
\end{example}

Throughout, we use \emph{sealing predicates}
$\varphi_{\seal}(X)$, which commit that the set of atoms $X$ is complete:
no additional atoms in $X$ will be derived in any admissible outcome.
Such predicates rule out outcomes and therefore constitute genuine semantic
commitments.

\subsection{The Canonical Layered Choice Basis}
\label{app:datalog-general}

\paragraph{Setup.}
Let $P$ be a finite $\mbox{Datalog}^\neg$ program with $k$
stratification layers and $d$ negative SCCs in the longest dependency
chain after SCC condensation.
Let $C_1, \ldots, C_d$ be a topological ordering of the negative SCCs
(if $C_i$'s resolution affects $C_j$'s admissible models, then
$i < j$).
SCCs that are incomparable in the DAG (no dependency path between them)
have choice predicates that do not interact: neither's rules reference
the other's atoms.
Their choices therefore commute and are placed in the same
determination layer.
The nesting depth $d$ counts the longest chain of dependent SCCs, not
the total number; independent SCCs at the same DAG level share a layer.

\paragraph{Relationship to the Gelfond-Lifschitz reduct.}
In the original formulation~\cite{gelfond1988stable}, a candidate model
$I$ is \emph{stable} if $I$ equals the minimal model of the
Gelfond-Lifschitz (GL) reduct $P^I$---the positive program obtained by
(i)~dropping every rule whose body contains $\neg C$ with $C \in I$,
and (ii)~deleting all remaining negated literals.
This is a global, verify-and-check definition.
Our layered choice basis reformulates it constructively: once atoms
outside a negative SCC $C_j$ are fixed (by the stratified prefix and
earlier SCC resolutions), the GL reduct of $C_j$'s rules depends only
on those fixed atoms, and its minimal models are exactly the
\emph{locally stable extensions} that our choice predicates enumerate.

\begin{definition}[Canonical layered choice basis]
  \label{def:canonical-choice-basis}
  The commitment basis has $k$ sealing predicates
  $\varphi_{\seal}(S_1), \ldots, \varphi_{\seal}(S_k)$---where $S_i$
  is the set of atoms sealed at step~$i$, declaring their derivations
  complete and licensing negation of $S_i$-atoms in subsequent
  rules---followed by $d$ choice layers: layer $k{+}i$ contains choice
  predicates that select a local stable extension of $C_i$, applied
  after all layers $k{+}1, \ldots, k{+}i{-}1$ are discharged.
  Sealing predicates do not commute across dependent strata (sealing
  $S_i$ is a prerequisite for evaluating rules that negate
  $S_i$-atoms); within a single choice layer, choice predicates
  commute.
\end{definition}

\noindent
A resolving determination takes the form
\[
  D \;\triangleq\;
  \bigl(
  \varphi_{\seal}(S_1)
  \seq \cdots \seq
  \varphi_{\seal}(S_k)
  \bigr)
  \seq
  \varphi_{C_1} \seq \cdots \seq \varphi_{C_d},
\]
where $\varphi_{C_j}$ denotes the (commuting) set of choice predicates
that resolve SCC $C_j$.
The determination depth is $k{+}d$.

\begin{proposition}[Soundness and completeness]
  \label{prop:canonical-basis-sc}
  For finite $\mbox{Datalog}^\neg$ programs whose negative SCCs admit
  the layered choice decomposition of
  Definition~\ref{def:canonical-choice-basis} (each SCC has finitely
  many local stable extensions given its fixed context, and composing
  local extensions in topological order is sound and complete for
  global stable models), the canonical layered choice basis is sound
  and complete: each resolving determination produces a stable model,
  and every stable model is produced by exactly one resolving
  determination.
\end{proposition}
\begin{proof}
  The full GL reduct decomposes along the SCC DAG: once atoms outside
  $C_i$ are fixed, the reduct of $C_i$'s rules depends only on those
  fixed atoms.

  \emph{Soundness:} composing local stable extensions in topological
  order yields a global stable model (each local extension is a minimal
  model of its local reduct; their composition is a minimal model of
  the full reduct).

  \emph{Completeness:} given a stable model $M$, define $D_M$ by
  taking the shared sealing prefix
  $\varphi_{\seal}(S_1) \seq \cdots \seq \varphi_{\seal}(S_k)$
  followed by the choice $\varphi_{C_i} = M|_{C_i}$ at each layer
  $k{+}i$.
  Local minimality of $M|_{C_i}$ follows from global minimality of $M$
  (a smaller local model would give a smaller global one).
  The map $M \mapsto D_M$ is injective.
\end{proof}

\paragraph{Instantiation (Example~\ref{ex:datalog-negation}).}
With $k = 1$ (one seal: $S_1 = \{p\}$) and $d = 1$ (one SCC:
$\{r, s\}$), the two resolving determinations are:
\[
  D^{(r)}
  \;\triangleq\;
  \varphi_{\seal}(S_1)
  \seq
  \varphi_{r(c)=\mathbf{t}} \cdot \varphi_{s(c)=\mathbf{f}},
  \qquad
  D^{(s)}
  \;\triangleq\;
  \varphi_{\seal}(S_1)
  \seq
  \varphi_{s(c)=\mathbf{t}} \cdot \varphi_{r(c)=\mathbf{f}}.
\]
Depth $= k{+}d = 2$.

\subsection{Negation Semantics as Filtration Levels}

\begin{theorem}[General negation semantics as filtration levels]
  \label{thm:negation-general}
  For a finite $\mbox{Datalog}^\neg$ program satisfying the
  layered-choice decomposition of
  Proposition~\ref{prop:canonical-basis-sc}, with $k$ stratification
  layers and nesting depth $d$, the determination semiring under the
  canonical layered choice basis
  (Definition~\ref{def:canonical-choice-basis}) has depth $k{+}d$.
  The semantics correspondence is:
  \begin{enumerate}[nosep,label=(\alph*)]
    \item \emph{Stable-model semantics} reads $\mathcal{F}_{k+d}$: all
          choice layers discharged, yielding a two-valued model.
    \item \emph{Well-founded semantics} reads $\mathcal{F}_{k+d^*}$
          where $d^* \le d$ is the number of consecutive choice
          layers (starting from $k{+}1$) whose resolution is forced
          (uniquely determined by the preceding prefix).
          Atoms with full or empty support are robust (WFS's
          $\mathbf{t}$ and $\mathbf{f}$); atoms with partial support
          are contingent (WFS's $\mathbf{u}$).
    \item \emph{Stratified semantics} reads $\mathcal{F}_k$: defined
          only when no negative cycles exist ($d = 0$).
  \end{enumerate}
\end{theorem}

\begin{proof}
  \emph{Depth $k{+}d$.}
  The $k$ sealing layers are shared by all resolving determinations
  (uniqueness of stratified evaluation).
  Each choice layer $k{+}i$ resolves the atoms in $C_i$; since $C_i$'s
  rules reference atoms in earlier cycles, the choices available at
  layer $k{+}i$ depend on the outcomes of layers $k{+}1, \ldots,
  k{+}i{-}1$.
  Hence layers $k{+}1, \ldots, k{+}d$ do not commute in general, giving
  depth $k{+}d$.

  \emph{(a) Stable models.}
  Each resolving determination $D$ discharges all $k{+}d$ layers,
  producing a unique two-valued assignment---a stable model.
  Conversely, every stable model corresponds to a resolving determination.
  Hence stable models biject with $\mathcal{D}$, and stable-model
  semantics reads $\mathcal{F}_{k+d} = 2^{\mathcal{D}}$.

  \emph{(b) Well-founded semantics.}
  WFS applies the alternating fixpoint~\cite{vanGelder1991wellFounded}
  after the sealing prefix.
  At each choice layer $k{+}i$, the alternating fixpoint checks whether
  $C_i$ has a unique stable extension given the discharged prefix.
  If so, that layer is \emph{forced} and WFS discharges it.
  If not, $C_i$'s atoms are classified as $\mathbf{u}$ and WFS stops
  discharging further layers that depend on $C_i$.
  Let $d^*$ be the number of consecutive forced layers starting from
  $k{+}1$.
  WFS reads $\mathcal{F}_{k+d^*}$: atoms resolved by layers
  $1, \ldots, k{+}d^*$ have full or empty support ($\mathbf{t}$ or
  $\mathbf{f}$); atoms depending on layers $k{+}d^*{+}1, \ldots, k{+}d$
  have partial support ($\mathbf{u}$).
  In the restricted case of Theorem~\ref{thm:negation-filtration}
  ($d = 1$, independent cycles), $d^* = 0$ and WFS reads $\mathcal{F}_k$.

  \emph{(c) Stratified semantics.}
  When $d = 0$, the sealing prefix resolves all atoms.
  $|\mathcal{D}| = 1$ and $\mathcal{F}_k = \{\emptyset, \mathcal{D}\}$.
\end{proof}

\paragraph{The alternating fixpoint as entailment.}
The distinction between the alternating fixpoint and a commitment is:
a \emph{commitment} operates on $2^O$ (replacing the admissible set
with a subset); the \emph{alternating fixpoint} operates on $O$ via
$\Ord$ (refining the observed outcome upward without excluding
alternatives).
It is the specification's choice of $\Ord$ that makes the alternating
fixpoint an entailment rather than a commitment.

\paragraph{Provenance structure.}
For atoms determined by the stratified prefix, all resolving
determinations agree: conditioned provenance is classical.
For atoms in negative SCCs, determination provenance records which
resolving determination is assumed; only after this choice does
ordinary derivational provenance apply.
WFS does not select a determination; it reports contingent atoms as
$\mathbf{u}$.
But the determination semiring retains the full contingency structure:
the support records exactly which determinations make each atom hold.
This is not provenance---provenance requires resolution
(Theorem~\ref{thm:resolution-implies-monotone})---but a property
\emph{of} the determination provenance function, read at filtration
level~$k$ without selecting a determination.

\begin{example}[Nested cycles]
  Consider: $a \leftarrow \neg b.\; b \leftarrow \neg a.\;
  c \leftarrow a, \neg d.\; d \leftarrow \neg c.$
  with EDB $\emptyset$.
  The $(a,b)$ cycle is at layer $k{+}1$; the $(c,d)$ cycle is at layer
  $k{+}2$ (it depends on $a$).
  When $a = \mathbf{t}$, the $(c,d)$ cycle has two stable extensions
  ($\{a,c\}$ and $\{a,d\}$); when $b = \mathbf{t}$, $c$ cannot be
  derived so $d = \mathbf{t}$ is forced---giving three stable models
  total.
  WFS: the alternating fixpoint cannot force $(a,b)$ (both extensions
  are consistent), so $d^* = 0$ and all four atoms are $\mathbf{u}$.
  If we add the rule $a \leftarrow.$ (making $a$ true by derivation),
  then $(a,b)$ is forced ($a = \mathbf{t}$, $b = \mathbf{f}$), so
  $d^* = 1$; WFS then evaluates $(c,d)$ and finds it unforced, giving
  $c = d = \mathbf{u}$.
\end{example}

\subsection{Worked Instantiation (Example~\ref{ex:datalog-negation})}

We trace the three semantics through the running example ($k=1$,
$d=1$, $\mathcal{D} = \{D^{(r)}, D^{(s)}\}$).

\paragraph{Stable models.}
The two stable models are $o_r = \{a, p, r\}$ and $o_s = \{a, p, s\}$.
Provenance:
$P_{D^{(r)}}(r(c)) = 1$, $P_{D^{(s)}}(s(c)) = 1$,
$P_{D^{(r)}}(p(c)) = P_{D^{(s)}}(p(c)) = x_a$
(the EDB annotation of $a(c)$).

\paragraph{Well-founded semantics.}
The stratified prefix determines $p(c) = \mathbf{t}$ (robust),
$q(c) = \mathbf{f}$ (robustly absent).
The $\neg$-cycle atoms have partial support:
$P(r(c)) = \{(D^{(r)}, 1)\}$, $P(s(c)) = \{(D^{(s)}, 1)\}$.
WFS reports $r(c) = s(c) = \mathbf{u}$ (contingent).

Adding the rule $d(c) \leftarrow r(c);\; d(c) \leftarrow s(c)$ gives
$\mathrm{supp}(P(d(c))) = \{D^{(r)}\} \cup \{D^{(s)}\} = \mathcal{D}$:
full support, read off algebraically without iterative fixpoint.

\paragraph{Stratified semantics.}
The stratified fragment ($p, q$) is resolved by the sealing prefix
alone: $P_{D_{\mathrm{strat}}}(p(c)) = P(a(c))$.
The cycle $(r, s)$ is outside its scope.

\paragraph{Cross-semantics comparison.}
$p(c)$ has full support---it holds under every stable model, is robust
under WFS, and is derived by stratified evaluation.
$r(c)$ has partial support ($\{D^{(r)}\}$): it holds under one stable
model, is contingent ($\mathbf{u}$) under WFS, and is undefined under
stratified semantics.

\subsection{Proof of Monus Elimination (Theorem~\ref{thm:monus-elimination})}

Let $P$ be a finite stratified $\mbox{Datalog}^\neg$ program with
stratification $S_0, S_1, \ldots, S_k$ (where $S_0$ is the EDB
stratum), and let $(K, +, \cdot, 0_K, 1_K)$ be a naturally ordered,
zero-divisor-free, $\omega$-continuous commutative semiring.
The zero-divisor-free condition ($a \cdot b = 0_K$ implies $a = 0_K$
or $b = 0_K$) is satisfied by all semirings commonly used in
provenance: $\mathbb{N}[X]$ (the provenance polynomials of Green et
al.~\cite{green2007provenance}), $\mathrm{PosBool}[X]$, the Viterbi
semiring, the tropical semiring, and the access-control
semiring---but not by semirings with absorbing elements such as
$(\{0,1,\ldots,k\}, \max, \min, 0, k)$ for finite $k > 1$.
The $\omega$-continuity condition (directed suprema distribute over
$+$ and $\cdot$) is the standard assumption for datalog fixpoint
semantics over
semirings~\cite{green2007provenance,dannert2021semiring}; it holds for
all the semirings above.
We prove that for every derived atom $a$, the support of $a$
(the set of determinations under which $a$'s annotation is nonzero)
is the same whether negation is handled via sealing or via monus.

\paragraph{Setup.}
Fix a single resolving determination $D$.
Under $D$, the specification is resolved: there is a unique outcome
(a model of the program), and classical semiring
provenance~\cite{green2007provenance} assigns each atom $a$ an
annotation $P_D(a) \in K$ recording how $a$ is derived from EDB facts.
The question is how to compute $P_D(a)$ when the program contains
negation.
Both strategies below produce a $K$-annotation for every atom; they
differ only in how negated literals are handled.
We compare them and show they agree on the zero/nonzero distinction.

\paragraph{Definitions.}
We recall the two evaluation strategies.
Both proceed bottom-up through the stratification, computing a least
fixpoint at each stratum over $(K, +, \cdot)$.
EDB atoms are annotated by their base semiring labels (or $0_K$ if
absent); positive body literals contribute their annotation
multiplicatively (as in standard $K$-relational
evaluation~\cite{green2007provenance}); alternative derivations
(multiple rules for the same head) combine additively.
The strategies differ only in the annotation assigned to a negated
literal $\neg b_j$ appearing in a rule body, where $b_j$ is in some
lower stratum whose fixpoint has already been computed.

\emph{Sealing evaluation.}
After stratum~$i$'s fixpoint stabilizes, \emph{seal} it.
At this point, standard $K$-relational
evaluation~\cite{green2007provenance} has assigned each atom $b$ in
stratum~$i$ a semiring value---call it $\mathrm{val}(b) \in K$---which
is a polynomial in the EDB annotations (in $\mathbb{N}[X]$) or more
generally an element of $K$ built from the base annotations via $+$
and $\cdot$.
This is the familiar provenance polynomial of $b$; it equals $0_K$
iff $b$ is not derivable from the EDB.
For a rule in stratum~$i{+}1$ containing $\neg b_j$ (with $b_j$ in
stratum~$\le i$, already sealed), the negated literal contributes:
\[
  \llbracket \neg b_j \rrbracket_{\mathrm{seal}}
  \;\triangleq\;
  \begin{cases}
    1_K & \text{if } \mathrm{val}(b_j) = 0_K
          \quad\text{(atom absent from the model)}, \\
    0_K & \text{if } \mathrm{val}(b_j) \neq 0_K
          \quad\text{(atom present)}.
  \end{cases}
\]
That is, a successful negation contributes $1_K$ (the multiplicative
identity---it does not block the derivation) and a failed negation
contributes $0_K$ (killing the rule's contribution).

\emph{Monus evaluation}~\cite{dannert2021semiring}.
The same bottom-up fixpoint computation, but the negated literal
contributes:
\[
  \llbracket \neg b_j \rrbracket_{\mathrm{monus}}
  \;\triangleq\;
  1_K \mathbin{\dot{-}} \mathrm{val}(b_j),
\]
where $\mathbin{\dot{-}}$ is the \emph{monus} (truncated subtraction)
of $K$, defined by
$a \mathbin{\dot{-}} b \triangleq \max\{\,c \in K \mid b + c \le a\,\}$
in the natural order $\le$ (where $x \le y$ iff $\exists d.\; x + d = y$).
Here $\mathrm{val}(b_j)$ is the same provenance polynomial as in the
sealing case---the two strategies compute the same fixpoint for
positive rules and differ only in how they interpret $\neg b_j$.

\paragraph{Key lemma.}
The only non-obvious step is showing that monus against $1_K$ acts as
a zero-test---the same behavior as sealing.
The tricky aspect of monus is that $1_K \mathbin{\dot{-}} v$ could in
principle be some nonzero value less than $1_K$ when $v \neq 0_K$; the
zero-divisor-free condition rules this out.
(For the Boolean support abstraction $\mathrm{PosBool}[X]$, where all
nonzero elements are $\ge 1_K$, the result is immediate; the lemma
below provides the stronger algebraic condition covering all standard
provenance semirings.)

\begin{lemma}
  \label{lem:monus-negation-value}
  In a naturally ordered, zero-divisor-free, $\omega$-continuous
  commutative semiring $K$:
  $1_K \mathbin{\dot{-}} v = 1_K$ if $v = 0_K$, and
  $1_K \mathbin{\dot{-}} v = 0_K$ if $v \neq 0_K$.
\end{lemma}
\begin{proof}
  The case $v = 0_K$ is immediate:
  $\max\{c \mid c \le 1_K\} = 1_K$.

  For $v \neq 0_K$: suppose for contradiction that $c \neq 0_K$
  satisfies $v + c \le 1_K$.
  Since $v, c \le 1_K$, monotonicity gives $vc \le v$ and $vc \le c$,
  so $vc + vc \le v + c \le 1_K$.
  Zero-divisor-freeness gives $vc \neq 0_K$.
  Squaring preserves both properties ($w \neq 0_K$ and $w + w \le 1_K$
  imply $w^2 \neq 0_K$ and $w^2 + w^2 \le 1_K$), so repeated squaring
  yields a descending chain $(vc)^{2^n}$, all nonzero, all satisfying
  $w + w \le 1_K$.
  By $\omega$-continuity this chain has an infimum
  $e = (vc)^{2^\infty}$, still nonzero, with $e^2 = e$ and
  $e + e \le 1_K$.
  We claim $e = 1_K$: write $e + f = 1_K$; if $f \neq 0_K$ then
  $ef \neq 0_K$ (zero-divisor-free) and $ef + ef \le e + f = 1_K$,
  so $ef$ is another nonzero element below $e$ satisfying the same
  bound---contradicting $e$ being the infimum.
  So $f = 0_K$ and $e = 1_K$.
  But then $1_K + 1_K \le 1_K$ implies $0_K \ge 1_K$, a contradiction.
\end{proof}

\paragraph{Main argument.}
Given the lemma, the two strategies assign the same zero/nonzero value
to every negated literal (both yield $1_K$ when the atom is absent,
$0_K$ when present).
Since positive evaluation is identical in both strategies, and $K$ is
zero-divisor-free (a product is nonzero iff all factors are; a sum is
nonzero iff some summand is), a straightforward induction on strata
shows that every atom has the same zero/nonzero provenance value under
sealing and monus.
Supports (the zero/nonzero distinction across determinations) therefore
agree.
\qed

\begin{remark}
  Under the assumptions of the theorem (zero-divisor-free,
  $\omega$-continuous), the lemma shows that monus and sealing assign
  \emph{identical} values to negated literals---both collapse to
  $\{0_K, 1_K\}$---so no quantitative information is lost.
  If one relaxes zero-divisor-freeness (e.g., the bounded semiring
  $(\{0,\ldots,k\}, \max, \min, 0, k)$), monus can assign intermediate
  values $1_K \mathbin{\dot{-}} v \notin \{0_K, 1_K\}$, encoding a
  quantitative residual that sealing discards.
  Whether this residual can be captured within the determination
  framework---perhaps via a richer commitment structure that records
  \emph{how much} of a derivation was overcome, not merely whether it
  was present---is an open question
  (Appendix~\ref{app:open-questions}).
\end{remark}

\section{Heredity Canonicalization}
\label{app:persistence-canonicalization}

A commitment is \emph{hereditary} if the outcomes it excludes stay
excluded as the history grows: once ruled out, an outcome never
becomes admissible again.
This is a natural monotonicity property---commitments should be
irrevocable---but not all commitment bases satisfy it syntactically.
For example, a commitment ``proposal $A$ wins the vote based on the
majority so far''
excludes outcomes where $A$ loses; but if new voters submit ballots, the
majority can flip, and the exclusion no longer holds.
The commitment's effect depends on context that has not yet stabilized.
We care about heredity because it guarantees that the admissible set
shrinks monotonically along any history extension, which simplifies
reasoning about determination structure: one need not worry that a
later event could ``undo'' an earlier commitment's exclusion.

We show that heredity is without loss of generality in the
retrospective setting: any non-hereditary commitment basis can be
transformed into a hereditary one (i.e., one where
$o \notin \Spec(H_1 \cdot \varphi) \Rightarrow o \notin \Spec(H_2 \cdot \varphi)$
for all $H_1 \sqsubseteq H_2$) that preserves not only resolved
outcomes but the full determination structure---the same set
$\mathcal{D}$, the same supports, the same filtration, and the same
$K$-valued provenance annotations.

\begin{definition}[Dependency set]
  The \emph{dependency set} of a commitment $\varphi$ at history $H$ is
  the set of events $S \subseteq E(H)$ whose presence or absence affects
  which outcomes $\varphi$ excludes.
  Formally, $S$ is minimal such that $\varphi$'s exclusion set is
  determined by $S \cap E(H)$ for all $H$.
\end{definition}

\begin{definition}[Sealing commitment]
  For a set of events $S$, the \emph{sealing commitment}
  $\varphi_{\seal(S)}$ declares that $S$ is complete: it excludes
  outcomes that are admissible only if additional events of the types in
  $S$ occur.
\end{definition}

\begin{proposition}[Heredity canonicalization preserves determination structure]
  \label{prop:persistence-canon}
  Let $\varphi$ be a non-hereditary commitment with dependency set $S$,
  applied at history $H$ within a completed history $H^*$ (so that $S$
  has stabilized: no further events of types in $S$ occur in any
  extension within $H^*$).
  Write $\Spec_{H^*}(H')$ for the admissible set at $H'$ conditioned
  on the completed history $H^*$ (i.e., restricted to outcomes
  consistent with $H^*$'s event structure).
  Define the replacement: $\varphi_{\seal(S)}$ followed by the
  deterministic function $f_\varphi$ that computes $\varphi$'s effect
  over the sealed state (an entailment, not a commitment).
  Then:
  \begin{enumerate}[nosep,label=(\alph*)]
    \item $\varphi_{\seal(S)}$ is hereditary.
    \item The seal is non-filtering at $H$ relative to $H^*$:
          $\Spec_{H^*}(H \cdot \varphi_{\seal(S)}) = \Spec_{H^*}(H)$
          (since $S$ has already stabilized within $H^*$, the seal
          excludes no outcome that was previously admissible).
    \item The entailment $f_\varphi$ reproduces $\varphi$'s exclusions
          exactly: $\Spec_{H^*}(H \cdot \varphi_{\seal(S)} \cdot f_\varphi)
          = \Spec_{H^*}(H \cdot \varphi)$.
    \item The remainder of the determination is unchanged: for any
          subsequent commitments $\varphi_2 \cdot \varphi_3 \cdots$,
          $\Spec_{H^*}(H \cdot \varphi_{\seal(S)} \cdot f_\varphi \cdot
          \varphi_2 \cdot \varphi_3 \cdots)
          = \Spec_{H^*}(H \cdot \varphi \cdot \varphi_2 \cdot \varphi_3 \cdots)$.
    \item The canonicalization preserves the set of resolving
          determinations $\mathcal{D}$, all supports, the filtration,
          and all $K$-valued provenance annotations.
  \end{enumerate}
\end{proposition}
\begin{proof}
  (a)~Once $S$ is declared complete, that declaration holds at all
  extensions (the seal event cannot be retracted).
  Any outcome excluded by sealing remains excluded: heredity holds.

  (b)~Since $H \sqsubseteq H^*$ and $S$ has stabilized within $H^*$,
  no further $S$-type events will arrive.
  The seal declares complete a set that is already complete; it excludes
  no outcome that was previously admissible.
  Hence $\Spec_{H^*}(H \cdot \varphi_{\seal(S)}) = \Spec_{H^*}(H)$.

  (c)~Non-heredity of $\varphi$ arose because its exclusion set
  depended on which events in $S$ had occurred.
  After sealing, $S \cap E(H)$ is fixed; $\varphi$'s effect is a
  deterministic function of this fixed set.
  The entailment $f_\varphi$ computes exactly this function, producing
  the same exclusion set as $\varphi$.
  Combined with (b): $\Spec_{H^*}(H) = \Spec_{H^*}(H \cdot \varphi_{\seal(S)})$,
  so $f_\varphi$ applied to $\Spec_{H^*}(H)$ gives $\Spec_{H^*}(H \cdot \varphi)$.

  (d)~By (c), $\varphi_2$ sees the same admissible set
  $\Spec_{H^*}(H \cdot \varphi)$ in both cases.
  Since $\varphi_2$'s effect depends only on the admissible set at its
  point of application, it behaves identically.
  By induction on the remaining commitments, the entire tail is unchanged.

  (e)~By (d), each resolving determination $D$ in the original basis
  maps to a determination in the canonicalized basis that produces the
  same admissible set at every point.
  The resolved outcome is identical; hence supports
  $\mathrm{supp}(P(t))$ are identical for every tuple $t$.
  Since supports are preserved, the filtration (defined over supports)
  is preserved.
  Since the resolved outcome under each $D$ is the same fixed instance,
  classical semiring provenance over that instance yields the same
  $K$-valued annotations.
\end{proof}

\noindent
The key insight is that in the retrospective setting, the seal is a
non-filtering guard (the dependency set has already stabilized, so no
outcome is excluded at the point of application),
and the entailment is not a commitment (it excludes nothing beyond what
$\varphi$ would have excluded).
The replacement is therefore invisible to the rest of the
determination: subsequent commitments see the same admissible set,
commutativity relations are unchanged, and the layer structure is
preserved.
Heredity is without loss of generality for the retrospective
analysis that determination provenance performs.

\section{Systems Directions}
\label{app:systems-lineage}
\label{app:certificates}

Algebraic data provenance~\cite{green2007provenance} offers powerful
compositional analysis---robustness, counterfactuals, quantitative
attribution---but applies only after semantics is fixed: a single
database, a single execution, a single model.
Systems observability (distributed
tracing~\cite{sigelman2010dapper}, workflow
provenance~\cite{seltzer2005provenance}) is broadly deployed but
expressively weak: it records which events occurred, not how they
combined or whether the outcome would survive a different execution.
Determination provenance bridges the two. Here we sketch some
opportunities and open problems.

\paragraph{From traces to algebraic provenance.}
Standard systems traces are histories in our sense: partially ordered
event records.
The only additional modeling step is to define an outcome space $O$ and
refinement order $\Ord$ (e.g., consistent serializations ordered by
sub-trace inclusion).
The commitment events---those that narrow the admissible set---are then
determined in principle, though identifying them efficiently requires
domain structure (conflict graphs, dependency annotations) rather than
brute-force comparison of $\Spec(H)$ before and after each event;
developing general-purpose extraction algorithms is an open problem.
The layer structure follows from the specification and the identified
commitments.

Once the determination is extracted, the full algebraic machinery
applies to the trace: supports give robustness and fragility;
the filtration gives depth bounds and layer-specific diagnosis;
responsibility attributes root causes to individual commitments;
cross-specification comparison answers ``what if'' questions about
alternative policies.
The framework applies whenever provenance queries can be expressed in
positive RA over the observation function---including semi-structured
settings (select, filter, join over trace spans or log entries).
The determination also compresses: rather than retaining every
intermediate state, the system stores only the commitment sequence and
reconstructs derivational detail on demand.

As a concrete example: in distributed
tracing~\cite{sigelman2010dapper}, unsynchronized clocks leave the
relative ordering of concurrent spans ambiguous.
A determination commits to clock alignments between hosts, fixing a
single timeline.
The support of a critical-path predicate then answers questions that
current tracing cannot: is the database \emph{robustly} on the
critical path (full support), or only under certain clock assumptions
(partial support)?
Or one could ask a counterfactual: if the clock alignment had gone
differently (a different determination in the same $\mathcal{D}$),
would the ad server rather than the database have been the
bottleneck?
The support $\mathrm{supp}(P(\mathrm{cp}=A))$ answers this directly.

\paragraph{Quantitative fragility and scheduler design.}
The support ratio $|\mathrm{supp}(P(t))| / |\mathcal{D}|$ measures
outcome fragility structurally---from the commitment basis and conflict
graph rather than from repeated execution.
Computing it exactly requires enumerating $\mathcal{D}$ (exponential in
general), but compact representations (Boolean formulas at depth~$1$;
factorized supports at higher depths) and the Shapley-value
approximation of Appendix~\ref{app:responsibility-details} make
tractable estimates feasible for bounded-treewidth conflict structures.
Given a determination structure extracted from prior runs---and
assuming the workload's conflict structure is stable across
runs---a scheduler could minimize maximum responsibility for an
SLA predicate by avoiding the high-responsibility commitments
(e.g., preferring orderings that keep fragile tuples robust).
This is a new optimization objective: optimize across the space of
determinations rather than within a single execution.

\paragraph{Parsimony: truncation and certificates.}
A full determination may record far more than any particular query
needs.
Two orthogonal dimensions of compression apply.

\emph{Vertical parsimony (truncation).}
The filtration tells us which layers matter: if a query family depends
only on tuples with $\mathrm{qdepth} \le k$, layers
$k{+}1, \ldots, d$ are irrelevant and can be discarded from the
stored determination.
For a depth-$d$ determination with $n$ commitments per layer, this
reduces the stored commitment sequence from $O(nd)$ to $O(nk)$.

\emph{Horizontal parsimony (certificates).}
Within a retained layer, the full commitment may carry more detail than
the query requires.
We envision a \emph{certificate} for commitment $\varphi$: a predicate $C$ that
over-approximates $\varphi$'s effect.
It must satisfy
$\Spec(H \cdot \varphi) \subseteq \{o \mid C(o)\}$
(soundness---every outcome that $\varphi$ retains is also retained by
$C$), but may admit additional outcomes that $\varphi$ would have
excluded.
A certificate is \emph{sufficient} for a query family $\mathcal{Q}$ if
replacing $\varphi$ with $C$ yields the same query answers.
The space savings come from the gap between the full commitment (which
may encode an entire conflict graph or the complete state of a
stratum) and the minimal certificate (which need only record the
aspects relevant to $\mathcal{Q}$---e.g., the ordering predicates in a
tuple's conflict neighborhood rather than the full serialization order).
Richer queries require stronger certificates:
derivational provenance needs only local conflict information, while
cross-specification queries need the full conflict-graph structure.

Together, truncation and certificates define a parameterized design
space for parsimonious determination storage; fixing the
parameters---choosing the query family, computing minimal certificates,
implementing adaptive retention---is the systems work that remains.

\section{Depth Reduction: Bounding Determination Complexity}
\label{app:depth-reduction}

The filtration measures determination complexity: a specification with
depth $d$ requires $d$ sequential, non-commuting commitment decisions
to resolve.
High depth means long chains of dependent decisions --- each contingent
on the previous --- and correspondingly high resolution cost.
A natural question: given a specification with depth $d$, how can we
reduce it to a target depth $d^* < d$?

We identify three mechanisms for depth reduction, each corresponding to
a different structural change to the specification or its commitment
basis.

\begin{definition}[Depth reduction mechanisms]
  \label{def:depth-reduction}
  Let $\Spec$ have depth $d$ under commitment basis $\Phi$.
  A \emph{depth reduction} to $d^* < d$ is achieved by one of:
  \begin{enumerate}[nosep,label=(\roman*)]
    \item \emph{Coarsening} $\Ord$: enlarge the outcome order to
          $\Ord' \supseteq \Ord$ so that outcomes distinguished by
          some layer~$k$ become $\Ord'$-comparable.
          Layer~$k$'s commitments become unnecessary (nothing to choose
          among) and the layer is eliminated.
    \item \emph{Commutation}: impose additional structure so that
          commitments in adjacent layers $k$ and $k{+}1$ commute.
          The two layers merge into one, reducing depth by one.
    \item \emph{Entailment}: replace a layer-$k$ commitment with a
          deterministic function of the discharged prefix (an
          entailment).
          The decision is no longer discretionary; the layer is
          eliminated.
  \end{enumerate}
\end{definition}

\noindent
Each mechanism has a cost:
(i)~weakens the specification (more outcomes are considered equivalent);
(ii)~restricts the commitment basis (decisions lose independence);
(iii)~removes discretion (the system commits to a policy).
The filtration diagnoses \emph{where} depth arises and \emph{which}
mechanism applies at each layer.

\begin{example}[Transactional depth reduction]
  \label{ex:depth-reduction-txn}
  The $T_\infty$ construction (Theorem~\ref{thm:txn-depth}) has depth
  $\Theta(n)$: each round's victim decision depends on the previous.
  The three mechanisms yield different bounded-depth protocols:
  \begin{itemize}[nosep]
    \item \emph{Coarsening} (SER $\to$ SI): the rw-rw cycle that
          forced victim selection is no longer forbidden under SI.
          The layer that broke the cycle is eliminated; depth drops for
          tuples involved in write-skew patterns.
    \item \emph{Commutation} (commutative basis rewriting): replace
          non-commutative commitments with globally commutative
          variants.
          For example, replace cycle-based victim selection
          ($\varphi_{\mathsf{abort}(T_v)}$, whose effect depends on
          which other aborts have already occurred---non-commutative)
          with per-transaction age-based abort predicates
          ($\varphi_{\mathsf{abort\text{-}if\text{-}old}(T_i)}$: abort
          $T_i$ if its age exceeds a threshold).
          The age predicates commute unconditionally (each depends only
          on its own transaction's metadata), so all abort decisions
          collapse into a single layer.
          The cost: the rewritten basis may abort transactions that
          precise victim selection would have spared---a semantic
          change, but one that buys global commutativity and reduces
          depth from $\Theta(n)$ to~$1$.
    \item \emph{Entailment} (reactive protocol): process each event
          deterministically given the current state (e.g., OCC:
          validate on commit, abort if cycle exists; MVTO: assign
          timestamp on arrival, abort on violation).
          No system choice arises; victim selection is entailed by the
          protocol and the arriving request.
          Depth drops to~$0$.
  \end{itemize}
\end{example}

\begin{example}[$\mbox{Datalog}^\neg$ depth reduction]
  \label{ex:depth-reduction-datalog}
  A program with $k$ stratification layers and $d$ nested negative
  SCCs has depth $k{+}d$ (Theorem~\ref{thm:negation-filtration}).
  \begin{itemize}[nosep]
    \item \emph{Coarsening} (stable $\to$ well-founded): enlarge
          $\Ord$ so that $\mathbf{u}$ refines to both $\mathbf{t}$ and
          $\mathbf{f}$.
          The choice layers become unnecessary (all outcomes are
          $\Ord'$-comparable); depth drops from $k{+}d$ to $k$.
    \item \emph{Commutation} (conservative choice): dependent
          SCC-resolution choices do not commute ($C_j$'s choice reads
          $C_i$'s outcome).
          Replace the precise choice for $C_j$ (minimal model of the
          GL reduct given $C_i$'s resolution) with a conservative one:
          include every atom in $C_j$ that is derivable under
          \emph{any} resolution of $C_i$.
          This choice no longer depends on $C_i$'s outcome, so the two
          layers commute and merge.
          The cost: the resulting model may not be minimal---it
          over-derives, including atoms that a precise reduct would
          have excluded.
          This parallels the transactional case (age-based abort
          over-aborts); both trade precision for commutativity.
          We present this merely as an illustration; whether such
          conservative models correspond to a natural
          semantics (perhaps supported models or some relaxation of
          stability) is unclear.
    \item \emph{Entailment} (deterministic choice): impose a policy
          that selects a unique stable extension for each negative SCC
          (e.g., order atoms alphabetically and assign $\mathbf{t}$ to
          the lowest in each cycle).
          The choice layer becomes deterministic --- an entailment, not
          a commitment.
          Each such SCC reduces $d$ by one.
  \end{itemize}
\end{example}

\section{Open Questions}
\label{app:open-questions}

Several structural questions about determinations remain open.

\paragraph{Algebraic structure of the determination space.}
Three related questions concern the internal algebra of determinations.
\emph{Equivalence:}
In classical trace theory~\cite{mazurkiewicz1977concurrent}, the Foata
normal form canonically represents a word in a partially commutative
monoid as a sequence of maximal independent layers.
The layered structure of a determination is the analog---but commutation
is \emph{dynamic}: whether $\varphi$ and $\psi$ commute depends on the
current admissible set, which changes as earlier layers are discharged.
Characterizing when two determinations are equivalent (produce the same
resolved specification) would give determinations a precise algebraic
identity beyond their definition.
\emph{Composition:}
When two determinations share a common prefix, they share initial
layers.
A formal treatment of determination morphisms---how determinations
relate under prefix extension, restriction, and substitution---would
enable compositional reasoning across sub-specifications.
\emph{Factorized representation:}
Within-layer commutativity gives $\mathcal{D}$ the structure of a trie
(each layer branches on the choice of commitment set; a determination
is a root-to-leaf path), yielding a factorized
representation~\cite{olteanu2015factorized} over the determination
space.
A level-$k$ support is a union of complete subtrees rooted at
depth~$k$, specified by its depth-$k$ prefixes without enumerating
exponentially many leaves.
This trie plays a role analogous to $\mathbb{N}[X]$ in classical
provenance: one computes the support trie once and reads off robustness,
qdepth, responsibility, or cross-spec comparison as different views.
But the analogy is imperfect: $\mathbb{N}[X]$'s universality is
algebraic (all specializations are semiring homomorphisms), whereas
our specializations include combinatorial operations (Shapley values)
and metric ones ($|\mathrm{supp}|/|\mathcal{D}|$) that are not
semiring maps.
Whether there is a formal category in which the support trie is
universal---and what efficient algorithms the factorized structure
enables---remain open.

\paragraph{Reversible commitments.}
In settings with rollback (speculative execution, backtracking search,
optimistic replication with conflict resolution), commitments can be
undone.
Extending the framework to reversible commitments---replacing monoids
with groups at each layer, in the spirit of incremental computation
via group-structured changes~\cite{budiu2023dbsp}---is open.

\paragraph{Depth under history extension.}
Under a hereditary basis
(Appendix~\ref{app:persistence-canonicalization}), the admissible set
can only shrink as history grows:
$\Spec(H') \subseteq \Spec(H)$ for all $H \sqsubseteq H'$.
One would expect that determination depth is also non-increasing---as
more events arrive, fewer commitments should be needed to resolve the
remaining ambiguity.
\emph{Conjecture:} under a hereditary basis, determination depth is
monotonically non-increasing under history extension.

\paragraph{Joint attribution.}
Classical provenance asks ``which input tuples are most responsible for
an output?''---a game where base tuples are players and the Shapley
value measures each tuple's marginal contribution to the
output~\cite{livshits2021shapley}.
Determination responsibility (Appendix~\ref{app:responsibility-details})
asks a different question: ``which commitments are most responsible for
this output holding under this resolution?''---a game where commitments
are players.
These two games are currently sequential: first attribute the
resolution (determination responsibility), then attribute the
derivation within that resolution (tuple-level responsibility).
Formalizing their composition into a single joint game---with both
commitments and tuples as players, and a unified Shapley value
measuring each player's contribution to the final output---would
unify the two dimensions of provenance into a single attribution.
The challenge is that the two player sets interact: a commitment can
make a tuple relevant (by including it in the resolved model), and a
tuple can make a commitment relevant (by participating in a conflict
that forces the commitment).

\paragraph{Program-bounded vs.\ input-recurrent depth.}
The $\Theta(n)$ transactional depth arises from recurring structure:
the same seal-plus-victim pair repeated as new conflicts arrive.
In $\mbox{Datalog}^\neg$, depth is bounded by the program's dependency
structure regardless of EDB size.
Characterizing when a specification's depth is bounded by its static
structure vs.\ when it grows with the input stream --- and whether
recurring depth admits uniform depth certification (a single check
covering all recurrences) --- is open.

\paragraph{Quantitative negation beyond zero-divisor-free semirings.}
Theorem~\ref{thm:monus-elimination} shows that sealing and monus agree
on supports under zero-divisor-free, $\omega$-continuous semirings.
In semirings that violate zero-divisor-freeness (e.g., bounded
lattice semirings $(\{0,\ldots,k\}, \max, \min, 0, k)$), monus can
assign intermediate values $1_K \mathbin{\dot{-}} v \notin \{0_K, 1_K\}$
to negated literals, encoding a quantitative residual---how much of a
derivation was ``overcome'' by the negation.
Sealing discards this information.
Whether the determination framework can be extended to capture such
quantitative negation---perhaps via graded commitments that record
residual strength rather than binary presence/absence---remains open.

\section*{Acknowledgments}
Generative AI tools (Claude and ChatGPT) were used as interactive
writing and reviewing assistants during the preparation of this work,
including drafting and revising text, checking mathematical arguments,
and suggesting edits.
All content was reviewed, validated, and approved by the authors.

\bibliographystyle{ACM-Reference-Format}
\bibliography{references}

\end{document}